\documentclass[%
 reprint,
 superscriptaddress,
 amsmath,amssymb,
 aps,
 physrev,
]{revtex4-2}

\usepackage{graphicx}
\usepackage{dcolumn}
\usepackage{bm}
\usepackage{mathrsfs}
\usepackage{tabularx}
\usepackage{pgfplots}
\pgfplotsset{compat=1.18}
\usepackage{amsthm}
\usepackage[colorlinks=true,allcolors=black]{hyperref}

\theoremstyle{definition}
\newtheorem{theorem}{Theorem}
\newtheorem{proposition}[theorem]{Proposition}
\newtheorem{lemma}[theorem]{Lemma}
\newtheorem{corollary}[theorem]{Corollary}
\newtheorem{axiom}[theorem]{Axiom}
\newtheorem{definition}[theorem]{Definition}
\newtheorem{example}[theorem]{Example}
\newtheorem{remark}[theorem]{Remark}

\newcommand{\norm}[1]{\lVert#1\rVert}
\newcommand{\abs}[1]{\left\lvert#1\right\rvert}

\begin{document} 

\preprint{APS/123-QED}

\title{A statistical theory of structure in many-particle systems with local interactions}

\author{John \c{C}amk{\i}ran}
\email{john.camkiran@utoronto.ca}
\affiliation{\mbox{Department of Materials Science and Engineering, University of Toronto, Toronto, Ontario, M5S 3E4, Canada}}

\author{Fabian Parsch}
\affiliation{\mbox{Department of Mathematics, University of Toronto, Toronto, Ontario, M5S 2E4, Canada}}

\affiliation{\mbox{Department of Materials Science and Engineering, University of Toronto, Toronto, Ontario,  M5S 3E4, Canada}}

\author{Glenn D. Hibbard}
\affiliation{\mbox{Department of Materials Science and Engineering, University of Toronto, Toronto, Ontario, M5S 3E4, Canada}}

\date{\today}

\begin{abstract}
A theory of structure is formulated for systems of many structureless classical particles with stable local interactions in Euclidean space. Such systems are shown to have their structure in thermodynamic equilibrium determined exactly by a random field of fine local descriptions and approximately by coarsenings thereof. The degree of order in the local cluster consisting of a particle and its neighbors is identified as a universal source of coarse local descriptions and characterized by expressing the behavior of configurational entropy in local microscopic terms. A local measure of the angular redundancy in neighboring particle positions is found to satisfy this characterization and thereby established as a valid local order quantifier. A precise relationship between order and symmetry is obtained by bounding this quantifier sharply from below by a simple function of the local point group and the largest stabilizer under its action on the set of bond pairs. The marginal distribution of the quantifier is given in closed form for highly coordinated particles with broadly distributed bond angles. Applications are made to the ideal gas, perfect crystal, and simple liquid.
\end{abstract}

\maketitle

\section{Introduction}
\label{sec:introduction}

More than a century after the advent of x-ray diffraction, our understanding of structure in condensed matter remains largely phase-specific. With at least one fundamental limitation on what properties can be evaluated by direct simulation \cite{Mehdi2024, Berthier2023, Bonati2021}, a broader theory of structure will be essential to furthering our ability to explain and predict the behavior of many-particle systems.

The early results of crystallography provide a special theory of structure grounded in the geometry of lattices and the algebra of discrete groups \cite{Aroyo2004, Giacovazzo2011, Bradley2010}. Reciprocal-space methods extend the scope of our understanding to noncrystalline systems \cite{Kittel2018, Fischer2006, Torquato2018} yet do so at the expense of spatial resolution \cite{Bernal1959, Sheng2006, Tanaka2019}. To address this shortcoming, such methods have been complemented with real-space descriptors to great practical effect \cite{Stukowski2012, Bartok2013, Royall2015a, Tanaka2019} but without substantial progress toward a general invariant of structure that is at once complete and tractable.

Indeed traditional invariants of structure are beset by the problem that they are either complete or tractable but not in general both. Perfect crystals at absolute zero are one setting in which the former can be attained without forgoing the latter. Outside that setting, the pair correlation function affords some degree of tractability \cite{Hansen2013,Pathria2011}. Yet in neglecting the higher-order dependencies that are often of importance \cite{Coslovich2013, Sharp2018, Zhang2020, Hu2022, Tanaka2025}, such two-particle descriptions lack completeness except in special cases \cite{Jiao2010, Stillinger2019, Maffettone2025}. And while the $n$-particle distribution function is trivially complete for sufficiently large $n$, its estimation suffers acutely from the curse of dimensionality \cite{NaglerCzado2016}.

The use of novel local descriptors has emerged as a promising alternative in the face of the dilemma inherent in traditional approaches \cite{Royall2015a, Tanaka2019}. Despite its empirical successes, however, this practice has yet to be placed on a firm theoretical foundation.

The present work arrives at a collection of results that formally concern systems of many structureless classical particles with stable local interactions but that nevertheless provide a perspective on all assemblies of countably many bodies described entirely by their species, positions, and momenta. Proceeding from classical theory, it is shown that the equilibrium structure of such systems can in principle be determined from little more than sufficiently rich local descriptions. Building on more recent findings, a definition for a local order quantifier is developed with reference to the behavior of configurational entropy as expressed in local microscopic terms. From there we derive a local measure of the angular redundancy in neighboring particle positions, which in satisfying that definition affords a natural means of describing structure at the length scale of interparticle interactions.

The remainder of the paper is organized as follows: Sec.~\ref{sec:structure-order} grounds the work in a treatment of structure and order in thermodynamic equilibrium; Sec.~\ref{sec:extracopularity} introduces a practical local order quantifier in the more general setting, establishes with its help a precise relationship between order and symmetry, and obtains an expression for its marginal distribution in closed form; Sec.~\ref{sec:elementary-materials} closes with a few elementary applications.

\section{Structure and order}
\label{sec:structure-order}

The extent to which the structure of a classical many-particle system can be captured with local descriptions is a question that remains unresolved notwithstanding the progress made \cite{Bernal1964, Baddeley1989, Steinhardt1983, Torquato1990, Boattini2020}. We address this problem in the following discussion by providing a sufficient condition under which local descriptions give rise to a complete structural invariant and consider the approximability of such an invariant when that condition is not strictly satisfied. Later we identify the local degree of order as a source of local descriptions that is universally available.

\subsection{Describing structure}
\label{subsec:describing-structure}

It is a familiar quality of systems with a large number of particles that nontrivial global structures can arise from strictly local interactions. In order to be able to discuss the description of structure in that setting, we must first make the notion of a local interaction precise. Throughout the work we call an interaction \textit{local} if it involves, in addition to a particle, one or more of its neighbors as given by some deterministic rule.

\subsubsection{Configurational probability density}
\label{subsubsec:configurational-probability-density}

Consider a system of $N$ structureless classical particles in Euclidean space $\mathbb{R}^D$ with interactions that are stable and local with respect to a symmetric irreflexive configuration-dependent and isometry-invariant neighbor relation $\sim$. Suppose that no external fields are present.

By its locality, every interaction is supported on a set of particles that induces a radius-$1$ subgraph of the $\sim$-neighbor graph $G(\mathcal{X})$, given for each configuration $\mathcal{X}$ by
\begin{equation*}
    G(\mathcal{X}) = \bigl( I, \, \bigl\{ \{i, j\} \in \tbinom{I}{2} : i \sim j \bigr\}\bigr),
\end{equation*}
where $I=\{1, \dots, N\}$ and $\tbinom{I}{2}$ is the set of all (unordered) pairs in $I$. A few radius-$1$ graphs are illustrated in Fig.~\ref{fig:radius-one-graphs}.

Denote by $\mathcal{R}(G(\mathcal{X}))$ the set comprising the vertex sets of all induced radius-$1$ subgraphs of $G(\mathcal{X})$ that support an interaction. The potential energy is then given by
\begin{equation}
    U(\mathcal{X}) = \sum_{R\in \mathcal{R}(G(\mathcal{X}))} \Psi(\mathcal{X}_R), 
\end{equation}
where $\mathcal{X}_R$ is the restriction of the configuration $\mathcal{X}$ to the set $R$ and $\Psi$ is an isometry-invariant interaction potential. By stability $U(\mathcal{X}) \geq CN$ for some $C \in (-\infty, 0]$ for all $\mathcal{X}$ \cite{Fisher1966}. We require also that $\Psi(\mathcal{X}_R) \geq c_R$ for some $c_R \in (-\infty, 0]$ for every $R$. Thus standard pair, multibody, and embedded-atom potentials are all possible $\Psi$.

If the system is in canonical equilibrium at temperature $T$ while confined to a region of volume $V$, then its configurational probability density function $f_X$ is given for all admissible configurations $\mathcal{X} \in \Gamma \subseteq (\mathbb{R}^D)^N$ by 
\begin{align}
    f_X(\mathcal{X}) &= \frac{1}{Z} \exp( -\beta U(\mathcal{X}) ) \nonumber \\
    &\propto  \exp\biggl\{-\beta \sum_{R \in \mathcal{R}(G(\mathcal{X}))} \Psi(\mathcal{X}_R) \biggr\}  \nonumber \\
    &=  \prod_{R \in \mathcal{R}(G(\mathcal{X}))} \exp ( -\beta \Psi(\mathcal{X}_R))   \nonumber \\
    &= \prod_{R \in \mathcal{R}(G(\mathcal{X}))} \psi(\mathcal{X}_R),
    \label{eq:boltzmann-factorization}
\end{align}
where $Z = \int_{\Gamma} \exp(-\beta U(\mathcal{X})) d\mathcal{X}$ is the canonical configurational partition function, $\beta = 1/(k_BT)$ is the inverse temperature, $k_B$ is the Boltzmann constant, and $\psi(\mathcal{X}_R): = \exp(-\beta\Psi(\mathcal{X}_R))$ is the Boltzmann factor.

\begin{figure}[t]
    \centering
    \includegraphics[width=\columnwidth]{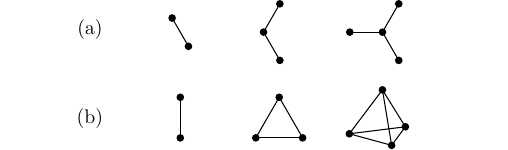}
    \caption{Radius-$1$ graphs include (a) star graphs and (b) complete graphs, but also wheel graphs, fan graphs, and others.}
    \label{fig:radius-one-graphs}
\end{figure}

\subsubsection{Fine local descriptions}
\label{subsubsec:fine-local-descriptions}

Structural descriptions are only truly local if they cannot be used to reconstruct the global configuration given the neighbor graph and are fine if they recover the Boltzmann factors. We formalize these notions below.

Let $w = (w_i)_{i \in I}$ be an indexed family of particle-centered descriptions for a configuration $\mathcal{X}$. The descriptions $w_i$ are \textit{local} if they are the images $g(\mathcal{X}_i)$ of the restricted configurations
\begin{equation*}
    \mathcal{X}_i := (x_j)_{j = i \,\text{or}\, j \sim i}
\end{equation*}
under a common function $g$ for all $\mathcal{X} \in \Gamma$ such that $(w, G(\mathcal{X}))$ does not in general determine $\mathcal{X}$ for connected $G(\mathcal{X})$ up to Euclidean motions. The descriptions $w_i$ are \textit{fine} if there exists a common function $\tau$ for all $\mathcal{X} \in \Gamma$ satisfying, for every $R \in \mathcal{R}(G(\mathcal{X}))$, the equality
\begin{equation}
    \tau(w_R) = \psi(\mathcal{X}_R), \quad w_R := (w_i)_{i \in R}.
    \label{eq:tau-map}
\end{equation}

\begin{figure}[b]
    \centering
    \includegraphics[width=\columnwidth]{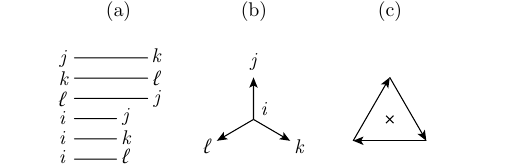}
    \caption{Three classes of structural descriptions: (a) labeled pairwise distances (fine local); (b) labeled neighbor vectors (fine nonlocal); (c) inter-neighbor vectors (non-fine local).}
    \label{fig:three-description-classes}
\end{figure}

An example of fine local descriptions is given by the sets of all labeled pairwise distances near each particle $i$,
\begin{equation*} 
    \mathcal{D}_i = \{ ( \{j,k\}, \norm{x_j - x_k}) : j \sim i, (k = i \text{ or } k \sim i), k \neq j \}. 
\end{equation*}
Indeed by classical distance geometry such descriptions recover, up to Euclidean motions, the restricted configuration $\mathcal{X}_i$ and therefore the Boltzmann factors $\psi(\mathcal{X}_R)$ but do not in general determine the global configuration $\mathcal{X}$. By contrast the sets of all labeled neighbor vectors,
\begin{equation*}
    \mathcal{V}_i = \{ ( j,  x_j - x_i )  : j \sim i \},
\end{equation*}
are fine but not local, since $((\mathcal{V}_i)_{i\in I}, G(\mathcal{X}))$ determines $\mathcal{X}$ up to a translation if $G(\mathcal{X})$ is connected. 
Fig.~\ref{fig:three-description-classes} illustrates these two classes of descriptions along with a third.

Together Eqs.~\eqref{eq:boltzmann-factorization} and \eqref{eq:tau-map} imply that
\begin{align}
    f_X(\mathcal{X}) 
    &= \frac{1}{Z} \prod_{R \in \mathcal{R}(G(\mathcal{X}))} \psi(\mathcal{X}_R) \nonumber \\
    &\propto \prod_{R \in \mathcal{R}(G(\mathcal{X}))} \tau(w_R) \nonumber \\
    &=: \mathcal{T}(w, G(\mathcal{X})).
\end{align}
In other words, for a fixed $(N, V, T)$ the density function $f_X$ and hence the full equilibrium structure of the system is determined (pointwise) by the random field $(g(X_i))_{i \in I}$ of local descriptions $g(X_i)$ on the random graph $G(X)$.

\subsubsection{Coarse local descriptions}
\label{subsubsec:coarse-local-descriptions}

Most practical local descriptions are not fine in the strict sense above. We therefore also consider the situation in which one has descriptions that are only approximately fine.

Let $w'$ be a random field of non-fine or \textit{coarse} local descriptions $w'_i = g'(X_i)$ with their $\sigma$-algebras $\sigma(w'_i) \subseteq \sigma(w_i)$, where each $w_i = g(X_i)$ is a fine local description. Then the $L^2$-optimal extension of $\tau$ to coarse local description tuples $w_R' := (w'_i)_{i \in R}$ satisfies
\begin{equation}
    \tau(w'_R) = \mathbb{E}[\psi(X_R) \mid \sigma(w'_R)] \quad \text{a.s.}
\end{equation}
for every $R \in \bigcup  \operatorname{supp}(\mathcal{R}(G(X)))$, where square integrability follows from the lower bound $\Psi(\mathcal{X}_R) \geq c_R > -\infty$. That is to say, for any $\sigma(w'_R)$-measurable square-integrable function $t$, one has
\begin{equation*}
    \mathbb{E}[(\psi(X_R) - \tau(w'_R))^2] \leq \mathbb{E}[(\psi(X_R) - t(w'_R))^2].
\end{equation*}

Consider a sequence of random fields of coarse local descriptions $w_i^{(n)}$ with their $\sigma$-algebras forming increasing filtrations $\sigma(w^{(n)}_i) \uparrow \sigma(w_i)$. We have $\sigma(w^{(n)}_R) = \bigvee_{i \in R} \sigma(w_i^{(n)})$ and hence $\sigma(w^{(n)}_R) \uparrow \sigma(w_R)$. Then
\begin{equation*}
    \tau(w^{(n)}_R) := \mathbb{E}[\psi(X_R) \mid \sigma(w^{(n)}_R)]
\end{equation*}
is a martingale with respect to the filtration $\sigma(w^{(n)}_R)$. 
By martingale convergence, for each fixed subset $R \subseteq I$,
\begin{align*}
    \tau(w^{(n)}_R) &= \mathbb{E}[\psi(X_R) \mid \sigma(w^{(n)}_R)] \\
    &\xrightarrow[n\to\infty]{a.s.} \mathbb{E}[\psi(X_R) \mid \sigma(w_R)].
\end{align*}
Moreover, by fineness, for every $R \in \bigcup  \operatorname{supp}(\mathcal{R}(G(X)))$, we have $\psi(X_R) =\mathbb{E}[\psi(X_R) \mid \sigma(w_R)]$ a.s.
Since there are finitely many subsets $R \subseteq I$, and a finite union of null sets is itself null, the product over $R \in \mathcal{R}(G(X)) \subset 2^I$ converges in like manner:
\begin{align}
    \prod_{R \in \mathcal{R}(G(X))} \tau(w^{(n)}_R)
    &\xrightarrow[n \to \infty]{a.s.} \prod_{R \in \mathcal{R}(G(X))} \psi(X_R) \nonumber \\
    &= \exp(-\beta U(X)) \nonumber \\
    &= Z f_X(X).
    \label{eq:integrand-convergence}
\end{align}

Define the normalizer
\begin{equation*}
    Z_n = \int_{\Gamma} \prod_{R \in \mathcal{R}(G(\mathcal{X}))}\tau(w^{(n)}_R) d\mathcal{X}.
\end{equation*}
From $\Psi(X_R) \geq c_R > -\infty$ we have $\psi(X_R) \leq \exp(-\beta c_R)$. The monotonicity of the conditional expectation gives
\begin{align*}
    0 &\leq \tau(w^{(n)}_R) \\
      &=\mathbb{E}[\psi(X_R) \mid \sigma(w^{(n)}_R)] \\
      &\leq \exp(-\beta c_R) \quad \text{a.s.}
\end{align*}
It follows that
\begin{equation*}
    \prod_{R \in \mathcal{R}(G(\mathcal{X}))}\tau(w^{(n)}_R) \leq \prod_{R \in \mathcal{R}(G(\mathcal{X}))}\exp(-\beta c_R) < \infty.
\end{equation*}

Because on $\Gamma$ we have $f_X > 0$ $d\mathcal{X}$-a.e., any statement about $X$ that holds a.s. also holds  $d\mathcal{X}$-a.e. on $\Gamma$. 

Since $N < \infty$, the bound $\prod_{R \in \mathcal{R}(G(\mathcal{X}))} \exp(-\beta c_R) \leq \prod_{R \subseteq I} \exp(-\beta c_R)$ is uniform on $\Gamma$. And since $V < \infty$ also, the product $\prod_{R \subseteq I} \exp(-\beta c_R)$ is integrable on $\Gamma$. Then by dominated convergence and Eq.~\eqref{eq:integrand-convergence},
\begin{align}
    Z_n 
    &= \int_{\Gamma} \prod_{R \in \mathcal{R}(G(\mathcal{X}))} \tau(w^{(n)}_R)  d\mathcal{X} \nonumber \\
    &\xrightarrow[n\to\infty]{} \int_{\Gamma} \prod_{R \in \mathcal{R}(G(\mathcal{X}))} \psi(\mathcal{X}_R) d\mathcal{X} \nonumber \\
    &= \int_\Gamma Z f_X(\mathcal{X})d\mathcal{X} \nonumber \\
    &= Z.
    \label{eq:normalizer-convergence}
\end{align}

Combining Eqs.~\eqref{eq:integrand-convergence} and \eqref{eq:normalizer-convergence} with the strict positivity of the normalizer $Z_n$, we arrive at
\begin{equation*}
    \frac{1}{Z_n} \mathcal{T}(w^{(n)}, G(\mathcal{X})) \xrightarrow[n \to \infty]{d\mathcal{X}\text{-a.e.}} f_X(\mathcal{X}).
\end{equation*}
Thus for a fixed $(N, V, T)$ the density function $f_X$ and hence the full equilibrium structure is determined approximately by the random field $(g^{(n)}(X_i))_{i \in I}$ of coarse local descriptions $g^{(n)}(X_i)$ on the random graph $G(X)$.

\subsection{Quantifying order}
\label{subsec:quantifying-order}

In the process of endowing a system with structure, interparticle interactions prevent its constituent particles from assuming statistically independent positions \cite{Torquato1990, Tanaka2012, Chiu2013}. This effect, known as order, is not peculiar to crystals as the frequent interchangeable use of the terms 
``noncrystalline'' and ``disordered'' might suggest but is also present in liquids and glasses, though in a much subtler form \cite{Steinhardt1983, Martin2002, Aste2005, Sheng2006, Tanaka2012, Tanaka2019}. In pure systems and ideal mixtures, order is strictly spatial, whereas more generally it may also be chemical \cite{Bragg1934, Cowley1950, Kikuchi1951, Sanchez1984}. Below we consider the degree of spatial order in the local cluster comprising a particle and its neighbors as a universal source of coarse local descriptions.

\subsubsection{The degree of order}
\label{subsubsec:the-degree-of-order}

In the study of phase transitions, order parameters serve to describe the spontaneous breaking of a Hamiltonian symmetry due to the onset of a specific kind of ordering \cite{Landau1980a, Sethna2021}. By contrast, order quantification seeks to report the total degree of order in a system as evidenced by its departure from complete randomness \cite{Chiu2013, Baddeley2015, vanLieshout2000}. Thus order quantifiers make natural order parameters, but the converse is not necessarily true. 

As a first step toward making the notion of a local order quantifier precise, we introduce the degree of order $\Omega$ via an axiom that generalizes to off-lattice many-particle systems the most basic properties of absolute net magnetization per site in zero-field Ising ferromagnets \cite{Landau1980b, Pathria2011}, which as illustrated in Fig.~\ref{fig:ising-ferromagnet} is perhaps the simplest setting in which structural order can arise.

\begin{axiom}
For all configurations $\mathcal{X}$, orthogonal matrices $Q$, and translation vectors $t$,
\par\medskip
\renewcommand{\arraystretch}{1.2}
\begin{tabular}{l@{\hspace{5pt}}l@{\hspace{16pt}}l}
    (a) & $\mathcal{X} \mapsto \Omega(\mathcal{X})$ & (microscopicity) \label{prop:microscopicity}  \\
    (b) & $0 \leq \Omega(\mathcal{X}) < \infty$ & (finite nonnegativity) \label{prop:nonnegativity}  \\
    (c) & $\Omega(\mathcal{X}) = \Omega(Q\mathcal{X} + t)$ &  (Euclidean invariance) \label{prop:invariance}
\end{tabular}
\label{ax:order-1}
\end{axiom}

\begin{figure}[b]
    \centering
    \includegraphics[width=\columnwidth]{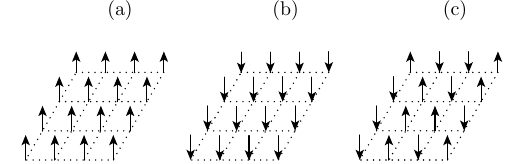}
    \caption{Absolute net magnetization per site in a zero-field Ising ferromagnet: In its ground state, the system assumes configuration (a) or its inversion (b) with probability one, giving $\abs{M}/N = 1$ (maximum order). Above the Curie point, each spin points in either direction with equal probability, giving $\langle \abs{M}/N \rangle = 0$ as typified in (c), with mutually independent spins as $T \to \infty$ (complete randomness).}
    \label{fig:ising-ferromagnet}
\end{figure}

Owing to the tendency of positional correlations to concentrate the configurational probability density, the degree of (dis)order in a system is closely related to its configurational entropy \cite{Nettleton1958, Baranyai1989, Angelani2007, Banerjee2014, Martiniani2016, Zhou2016, Yang2016, Piaggi2017, Hallett2018}. Indeed the behavior of configurational entropy in simple many-particle systems, such as bulk argon, provides a valuable macroscopic point of reference on the average microscopic degree of order. We capture this behavior in a second axiom, whenceforth we continue to write $N$ for the number of particles, $V$ for the volume, and $T$ for the temperature, but now also write $s$ for the configurational entropy per particle and $\langle \cdot \rangle$ for the ensemble average.

\begin{axiom}
For bulk argon in the single-phase region, away from the critical point,
\par\medskip
\renewcommand{\arraystretch}{1.2}
\begin{tabular}{l@{\hspace{5pt}}l@{\hspace{7pt}}l}
    (a) & $(\partial \langle\Omega\rangle / \partial s)_{N, T} < 0$ & (const-$T$ monotonicity in $s$) \\
    (b) & $(\partial \langle\Omega\rangle / \partial s)_{N, V} < 0$ & (const-$V$ monotonicity in $s$)
\end{tabular}
\label{ax:order-2}
\end{axiom}

In sum, Axioms~\ref{ax:order-1} and \ref{ax:order-2} assert that the degree of order in a many-particle system is given by a nonnegative real-valued function of the prevailing configuration with the additional properties that it is preserved under Euclidean motions and that its ensemble average in bulk argon decreases with increasing configurational entropy per particle at constant temperature or volume.

\subsubsection{Local order quantifier}
\label{subsubsec:local-order-quantifier}

It is possible to regard the local cluster consisting of a central particle $i$ and its neighbors $j \sim i$ as a system in its own right. Under this view, Axiom~\ref{ax:order-1}(a) implies that the degree of spatial order in a local cluster is prescribed by some function $\omega$ to be equal to $\Omega(\mathcal{X}_i)$, where $\mathcal{X}_i$ is the restriction of the configuration $\mathcal{X}$ to that cluster.

Let $\hat{\mathcal{V}}_i$ denote the set of all unlabeled neighbor vectors of the $i$th particle and consider its orbit under the action of the $D$-dimensional orthogonal group $O(D)$. Fix a choice of representative $c$ and call the chosen representative $B_i = c(O(D)\hat{\mathcal{V}}_i)$ a \textit{bond set} \footnote{One choice can be obtained via singular value decomposition and the Gram matrix of the neighbor vector matrix}. Then by Axiom~\ref{ax:order-1}(c) there exists a function $\Phi$ satisfying $\Omega(\mathcal{X}_i) = \Phi(B_i)$. From this and Axiom~\ref{ax:order-1}(b) it is immediate that
\begin{equation}
    B_i \mapsto \omega(B_i), \quad 0 \leq \omega(B_i) < \infty.
    \label{eq:little-omega}
\end{equation}

We check before proceeding that $\omega(B_i)$ is a valid local description per Sec.~\ref{subsubsec:fine-local-descriptions}. In representing the orthogonal orbit $O(D)\hat{\mathcal{V}}_i$ with the bond set $B_i$, we omit the orientations of the underlying neighbor vectors. Consequently the pair $((B_i)_{i \in I}, G)$ does not in general determine the global configuration $\mathcal{X}$, which implies that $B_i$ is a local description. Thus $\omega(B_i)$, being the image of a local description, must itself be a local description.

In experiment and simulation, the microscopic character of a local cluster is frequently found to be dominated by two geometric features: the number of bonds formed by the central particle or \textit{\mbox{coordination} number} $k$ \cite{Bernal1967, Torquato2000, Torquato2010, Zaccone2011, Schoenholz2016, Laurati2017, Xia2017, Mizuno2020, Bapst2020, Zaccone2022, Anzivino2023} and the angles between those bonds or \textit{bond angles} $\Theta$ \cite{Coslovich2011, Tanaka2012, Coslovich2013, Cubuk2015, Russo2016, Sharp2018, Tong2018, Harrington2019, Zhang2020, Boattini2021, Pozdnyakov2022}. Following these observations, we proceed with the localization of Axiom~\ref{ax:order-2}(a) and (b) by studying the ensemble behaviors of $k$ and $\Theta$.

Let us start by considering the quasistatic compression of liquid argon at constant temperature $T$. The configurational entropy $S$ of such a system \mbox{decreases} with an isothermal decrease in volume $V$ at constant $N$, so that
\begin{equation}
    \frac{\partial s}{\partial \rho} = \frac{1}{N} \frac{\partial S}{\partial V} \frac{\partial V}{\partial \rho} = -\frac{1}{\rho^2} \frac{\partial S}{\partial V} <0,
    \label{eq:entropy-density}
\end{equation}
where $\rho = N/V$ is the number density. 

What we seek below is the direction of the concomitant change in the ensemble-average coordination number $\langle k \rangle$.

Recall that the radial distribution function (RDF) is a nonnegative function $g$ given in the thermodynamic limit for radial distances $r > 0$ and number densities $\rho > 0$ by
\[
    g(r; \rho) = \lim_{\substack{N,V \to \infty \\ N/V = \rho}} \frac{V}{N(N-1)} \Biggl\langle \sum_{i} \sum_{j\neq i} \frac{\delta(r-\norm{x_i - x_j})}{4 \pi r^2} \Biggr\rangle.
\]
Under homogeneous and isotropic conditions such as here considered, $\langle k \rangle$ is obtained in that limit as follows:
\begin{equation*}
    \langle k \rangle = 4 \pi \rho \int_{0}^{r_1} r^2 g(r; \rho)  dr,
\end{equation*}
where $r_1$ is the $\rho$-dependent location of the first minimum of $g$. We reduce this to an integral over the unit interval by making the change of variables
\begin{equation*}
    y = r/r_1.
\end{equation*}

Let $Y_1 < 1$ be the reduced location of the first maximum. Define $\Lambda$ as a triangular function with apex at $Y_1$, left width $0 < W_L < Y_1$, and right width $W_R = 1 -Y_1$:
\begin{equation*}
    \Lambda(y) =
    \begin{cases}
        0 & y \leq Y_1- W_L \\
        1 - \dfrac{Y_1-y}{W_L} &  Y_1 - W_L < y \leq Y_1  \\[1.5ex]
        1 - \dfrac{y-Y_1}{W_R} &  Y_1 < y \leq Y_1 + W_R \\
        0 & y > Y_1 + W_R = 1.
    \end{cases}
\end{equation*}

Denote by $A(\rho) = g(Y_1 r_1; \rho)$ the height of the first RDF peak, which is observed empirically to satisfy
\begin{equation*}
    \frac{\partial A}{\partial \rho} > 0.
\end{equation*}
We may then model the first peak as
\begin{equation*}
    g_\Lambda(r; \rho) = A(\rho)\Lambda(r/r_1), \quad 0 \leq r \leq r_1,
\end{equation*}
and approximate the average coordination number by
\begin{align*}
    \langle k \rangle _\Lambda(\rho) 
    &= 4\pi\rho \int_0^{r_1} r^2 {g}_{\Lambda}(r; \rho) dr \\
    &= 4\pi\rho r_1^3\int_0^1 y^2 {g}_{\Lambda}(yr_1; \rho) dy \\
    &= 4\pi\rho r_1^3 A(\rho) C,
\end{align*}
where
\begin{align*}
    C &= \int_0^{1} y^2 \Lambda(y) dy \\
      &= \frac{W_L^3 + W_R^3}{12} + \frac{Y_1^2}{2}(W_L+W_R) + \frac{Y_1}{3}(W_R^2-W_L^2) > 0.
\end{align*}

We express the difference $\Delta$ between the true average $\langle k \rangle$ and model average $\langle k \rangle_\Lambda$ in terms of the error $\epsilon$ of the peak model $g_\Lambda$ with respect to the true first peak of $g$:
\begin{align*}
    \Delta(\rho) 
    &= \langle k \rangle (\rho) - \langle k \rangle_\Lambda(\rho) \\
    &= 4\pi \rho \int_0^{r_1} r^2 [g(r;\rho) - A(\rho)\Lambda(r/r_1)]dr \\
    &= 4\pi \rho r_1^3 \int_0^1 y^2 [g(yr_1; \rho) - A(\rho)\Lambda(y)] dy \\
    &= 4\pi \rho r_1^3 A(\rho) \int_0^1 y^2 \biggl(\frac{g(yr_1; \rho)}{A(\rho)} - \Lambda(y)\biggr) dy \\
    &= 4\pi \rho r_1^3 A(\rho) \underbrace{\int_0^1 y^2 \epsilon(y; \rho) dy}_I.
\end{align*}

Over small changes in density, interparticle distances scale to leading order as the inverse cube root of density, so that $r_1 \approx c\rho^{-1/3}$ for some $c>0$. Insofar as the triangle $\Lambda(y)$ agrees with the $A$-scaled first peak $g_\Lambda(yr_1; \rho)/A(\rho)$, the error $\epsilon$ and hence the integral $I$ will be small. And since quasistatic changes in the density of a simple liquid result primarily in the horizontal and vertical dilation of the first RDF peak, $\epsilon$ and $I$ will also be slow with respect to $\rho$, rendering $\partial I/\partial \rho$ likewise small. Thus
\begin{align}
    \frac{\partial \langle k \rangle}{\partial \rho}
    &= \frac{\partial}{\partial \rho}(\langle k \rangle_\Lambda(\rho) + \Delta(\rho)) \nonumber \\
    &= 4\pi \frac{\partial}{\partial \rho}(\rho r_1^3 A(\rho)\bigl(C+I)) \nonumber \\
    &= 4\pi \biggl\{ A(\rho)\bigl(C+I\bigr)\frac{\partial}{\partial \rho}\bigl(\rho r_1^3\bigr) \nonumber \\
    &\hphantom{=4\pi \biggl\{}\, +\rho r_1^3\biggl(\frac{\partial A}{\partial \rho}\bigl(C+I\bigr)
    + A(\rho)\frac{\partial I}{\partial \rho}\biggr)\biggr\} \nonumber \\
    &\approx 4\pi c^3\biggl(\frac{\partial A}{\partial \rho}\bigl(C+I\bigr)
    + A(\rho)\frac{\partial I}{\partial \rho}\biggr) \nonumber \\
    &= 4\pi c^3\biggl( C \frac{\partial A}{\partial \rho} + \frac{\partial A}{\partial \rho} I + A(\rho)\frac{\partial I}{\partial \rho}\biggr) \nonumber \\
    &\approx 4\pi c^3 C \frac{\partial A}{\partial \rho} \nonumber \\
    &> 0.
    \label{eq:coordination-density}
\end{align}

Ineqs.~\eqref{eq:entropy-density} and ~\eqref{eq:coordination-density} together with Axiom~\ref{ax:order-2}(a) imply that the ensemble-average degree of spatial order in a local cluster increases with the ensemble-average coordination number of its central particle:
\begin{equation}
    \frac{\partial \langle \omega \rangle}{\partial \langle k \rangle} = \frac{\partial \langle \Omega \rangle}{\partial s} \frac{\partial s}{\partial \rho} \biggl(\frac{\partial \langle k \rangle}{\partial \rho}\biggr)^{-1} > 0. 
    \label{ineq:average-coordination} \tag{A}
\end{equation}

Let us now consider the quasistatic heating of crystalline argon at constant volume $V$. In canonical equilibrium at $T>0$ the constant-volume specific heat capacity $c_V$ of any system with more than one energetically distinct accessible microstate satisfies
\begin{align}
    c_V = T\frac{\partial s}{\partial T} > 0.
    \label{ineq:heat-capacity}
\end{align} 

Up to very high pressures, solid argon is known to have an fcc crystal structure \cite{Dewaele2021}. With the coordination number of the central particle thus fixed at $k=12$, our attention is turned to its bond angles, each belonging to one of $m=4$ zero-temperature classes: $60^\circ$, $90^\circ$, $120^\circ$, $180^\circ$.

Every marginal of a Gibbs measure can be written in exponential form with respect to an effective potential or potential of mean force (PMF) \cite{Chandler1987}. The partial bond angle distribution $f_{\Theta_i}$ corresponding to the $i$th angle class is given by
\begin{equation*}
    f_{\Theta_i}(\vartheta) = \frac{\exp(-\beta U_{\Theta_i}(\vartheta; T))}{Z_{\Theta_i}(T)},
\end{equation*}
where $U_{\Theta_i}$ is the PMF along the $i$th angular coordinate $\Theta_i$ and $Z_{\Theta_i}$ is the corresponding normalizer.

The diversity of the bond angles in each class can be assessed by the bond angle entropy $h(\Theta_i)$, evaluated as
\begin{equation*}
\begin{aligned}
    h(\Theta_i) 
    &= -\int f_{\Theta_i}(\vartheta) \ln (f_{\Theta_i}(\vartheta)) d\vartheta \\
    &= -\int f_{\Theta_i}(\vartheta) [-\beta U_{\Theta_i}(\vartheta; T) - \ln(Z_{\Theta_i}(T))] d\vartheta \\
    &= \beta\int \! U_{\Theta_i}(\vartheta; T) f_{\Theta_i}(\vartheta) d\vartheta + \ln(Z_{\Theta_i}(T))\!\int \!\!f_{\Theta_i}(\vartheta) d\vartheta \\
    &= \beta \langle U_{\Theta_i} \rangle + \ln(Z_{\Theta_i}(T)).
\end{aligned}
\end{equation*}
Differentiating this with respect to temperature yields
\begin{align*}
    \frac{\partial h(\Theta_i)}{\partial T} 
    &= \frac{\partial h(\Theta_i)}{\partial \beta}\frac{\partial \beta}{\partial T} \nonumber \\ 
    &= \biggl(\beta \frac{d \langle U_{\Theta_i} \rangle}{d \beta} - \beta \biggl\langle \frac{\partial U_{\Theta_i}}{\partial \beta} \biggr\rangle \biggr) \biggl( - \frac{1}{k_B T^2} \biggr) \nonumber \\
    &= \frac{1}{k_B^2 T^3} \biggl( \biggl\langle \frac{\partial U_{\Theta_i}}{\partial \beta} \biggr\rangle - \frac{d \langle U_{\Theta_i} \rangle}{d \beta} \biggr) \nonumber \\ 
    &> 0
\end{align*}
in the following two cases.

At low $T>0$, small thermal displacements from the fcc arrangement place the crystal in the harmonic regime, in which each angular PMF is separable as $U_\Theta(\vartheta; T) = \mathcal{U}_\Theta(\vartheta) +  C(T) + \Delta(\vartheta; T)$ with $\Delta(\vartheta; T) \approx 0 $ and $\partial \Delta / \partial \beta \approx 0$. We have in this case that
\begin{align*}
    \biggl\langle \frac{\partial U_\Theta}{\partial \beta} \biggr\rangle - \frac{d \langle U_\Theta \rangle}{d \beta}
    &= C'(T) + \operatorname{Var}(\mathcal{U}_\Theta) - C'(T)\\
    &= \operatorname{Var}(\mathcal{U}_\Theta)\\
    &> 0.
\end{align*}
At higher $T$ away from the melting point $T_m$, 
\begin{align*}
    \biggl\langle \frac{\partial U_\Theta}{\partial \beta} \biggr\rangle - \frac{d \langle U_\Theta \rangle}{d \beta}
    &= \operatorname{Var}(U_\Theta)+\beta\operatorname{Cov}(U_\Theta,\partial_\beta U_\Theta) \\
    &\ge \operatorname{Var}(U_\Theta)-\beta\sqrt{\operatorname{Var}(U_\Theta)\operatorname{Var}(\partial_\beta U_\Theta)} \\
    &= \operatorname{Var}(U_\Theta) \Biggl(1-\beta\sqrt{\frac{\operatorname{Var}(\partial_\beta U_\Theta)}{\operatorname{Var}(U_\Theta)}}\Biggr) \\
    &= \operatorname{Var}(U_\Theta) \Biggl(1-\sqrt{\frac{\operatorname{Var}(T S_\Theta)}{\operatorname{Var}(U_\Theta)}}\Biggr) \\
    &> 0,
\end{align*}
where the first relation is due to the identity $d\langle U\rangle /d\beta =\langle \partial U /\partial \beta \rangle-\operatorname{Var}(U)-\beta\operatorname{Cov}(U,\partial_\beta U)$ \cite{Tuckerman2010}; the second follows from the Cauchy--Schwarz inequality; the third inserts $S_\Theta := -{\partial U_\Theta(\vartheta; T)}/{\partial T}$; and the fourth holds if $\operatorname{Var}(T S_\Theta)<\operatorname{Var}(U_\Theta)$, 
as one is led to expect from the leading-order behaviors of $U_\Theta$ and $TS_\Theta$,
\begin{align*}
    U_\Theta(\vartheta; T) &\approx U_\Theta(\vartheta_0; T) + \frac{1}{2} \kappa (\vartheta_0)(\vartheta - \vartheta_0)^2, \\
    T S_\Theta(\vartheta; T) &\approx T S_\Theta(\vartheta_0; T) + \frac{k_B T}{2} \ln\mathopen{}\biggl( \frac{\kappa ( \vartheta)}{\kappa(\vartheta_0)} \biggr);
\end{align*}
where $\kappa(\vartheta) = \partial^2 U_\Theta / \partial \vartheta^2$ is the effective stiffness.

Away from $T_m$, the angle classes are nearly disjoint, so that the class-wise result carries over to the full bond angle distribution $f_\Theta := \sum_i \alpha_i f_{\Theta_i}$ with error $\epsilon(T)$:
\begin{align}
    \frac{\partial h(\Theta)}{\partial T}
    &= \frac{\partial}{\partial T} \biggl( \sum_{i=1}^m \alpha_i h(\Theta_i) + \sum_{i=1}^m \alpha_i \ln \alpha_i - \epsilon(T) \biggr) \nonumber \\
    &\approx \sum_{i=1}^m \alpha_i \frac{\partial h(\Theta_i)}{\partial T} \nonumber \\
    &> 0, \label{ineq:temperature-dependence}
\end{align}
where the exact equality follows from the entropy of a disjoint mixture \cite[Lemma~2.5.2]{Gray2011} and the approximate equality from $\partial \alpha_i / \partial T = 0$ (fixed class weights) and $\partial \epsilon / \partial T \approx 0$ (slowly increasing error due to class overlaps).
    
The combination of Axiom~\ref{ax:order-2}(b) with Ineqs.~\eqref{ineq:heat-capacity} and \eqref{ineq:temperature-dependence} indicates that the ensemble-average degree of spatial order in a local cluster decreases as the bond angles of the central particle become more diverse:
\begin{align}
    \frac{\partial \langle \omega \rangle}{\partial h(\Theta)} &= \frac{\partial \langle \Omega \rangle}{\partial s} \frac{\partial s}{\partial T}
    \frac{\partial T}{\partial h(\Theta)} \nonumber \\ 
    &= \frac{c_V}{T}\frac{\partial \langle \Omega \rangle}{\partial s}
    \biggl(\frac{\partial h(\Theta)}{\partial T}\biggr)^{-1} \nonumber \\ &< 0. 
    \label{ineq:bond-angle-entropy} \tag{B}
\end{align}

Ineqs.~\eqref{ineq:average-coordination} and \eqref{ineq:bond-angle-entropy} recast Axiom~\ref{ax:order-2}(a) and (b) in terms of the ensemble behaviors of the local features $k$ and $\Theta$ respectively. All that remains is to consider the microscopic contributions to these effects.

It can be shown that an increase in the average degree of spatial order $\langle \omega \rangle$ due  to an increase in the average coordination number $\langle k \rangle$ is implied by an increase in the degree of spatial order $\omega$ due to an increase in the coordination number $k$:
\begin{equation}
    \frac{\partial \omega}{\partial k} > 0 \quad \Rightarrow \quad \frac{\partial \langle \omega \rangle}{\partial \langle k \rangle} > 0.
\end{equation}

It can likewise be shown that a decrease in $\langle \omega \rangle$ due to an increase in the ensemble bond angle entropy $h(\Theta)$ is implied by a decrease in $\omega$ due to an increase in the microscopic bond angle entropy $H(\theta)$, defined in Sec.~\ref{subsec:parameter}:
\begin{equation}
    \frac{\partial \omega}{\partial H(\theta)} < 0 \quad \Rightarrow \quad \frac{\partial \langle \omega \rangle}{\partial h(\Theta)} < 0.
\end{equation}

With both Axioms~\ref{ax:order-1} and \ref{ax:order-2} now expressed in local microscopic terms, we summarize the derived properties in the following definition.

\begin{definition}[local order quantifier]
Given a particle with bond set $B$, a \textit{local order quantifier} $\omega$ is a scalar-valued function that has the following properties: 
\par\medskip
\renewcommand{\arraystretch}{1.2}
\begin{tabular}{l@{\hspace{5pt}}l@{\hspace{25pt}}l}
    (1) & $B \mapsto \omega(B)$ & (local microscopicity) \label{cond:loq-1} \\
    (2) & $0 \leq \omega(B) < \infty$ & (finite nonnegativity) \label{cond:loq-2} \\
    (3) & $\partial \omega / \partial k > 0$ & (monotonicity in $k$) \label{cond:loq-3} \\
    (4) & $\partial \omega / \partial H(\theta) < 0$ & (monotonicity in $H(\theta)$) \label{cond:loq-4}
\end{tabular}
\label{defn:local-order-quantifier}
\end{definition}

Definition~\ref{defn:local-order-quantifier} captures in a few generic properties the geometry of spatial order in a small cluster of particles. In so doing it supplies a set of criteria for the appraisal of the local degree of order outside the setting of simple systems in thermodynamic equilibrium. 

\section{Extracopularity}
\label{sec:extracopularity}

In view of the microscopic nature of spatial order, the study of particle positions presents itself as a logical point of departure in the search for a practical local order quantifier. In the ensuing discussion, we derive a local measure of the redundancy in neighboring particle positions and find it to be a local order quantifier in the sense of Definition~\ref{defn:local-order-quantifier}. We then use this quantifier to establish a relationship between order and symmetry and later obtain an approximate expression for its marginal distribution in closed form.

\begin{remark}
Here and below the term ``entropy'' will refer exclusively to the information-theoretic quantity introduced by Shannon \cite{Shannon1948}, which provides the information content of a data source in units of bits \cite{Cover2006, Gray2011}. The information-theoretic and thermostatistical notions of entropy are closely related: in the lattice formulation, Gibbs entropy is equal to the (discrete) Shannon entropy of the microstate probabilities up to a factor of $k_B \ln2$, and in the continuum formulation, to the differential (Shannon) entropy of the phase space density up to a factor of $k_B$ \cite{Sethna2021}.
\end{remark}

\subsection{Parameter}
\label{subsec:parameter}
 
For the vast majority of configurations, a complete multiset of exact pairwise distances suffices to determine the positions of $N \geq D+2$ particles in $D$ dimensions, up to congruence and relabeling \cite[Theorem 1.6]{Boutin2004} (for a discussion on why the pair correlation function is nevertheless an incomplete invariant of structure, see Refs.~\onlinecite{Jiao2010, Stillinger2019}). Such distances serve as the basis for the derivation that follows.

\subsubsection{Derivation}
\label{subsubsec:derivation}

We begin by observing that there are really only two kinds of distances in the cluster comprising a particle and its neighbors: those between the central particle and its neighbors and those among its neighbors. 

Notice that each distance of the former kind corresponds to the length of a  bond and each of the latter kind to the length of the difference between two distinct bonds. In a set of $k$ bonds one has $n=\tbinom{k}{2}$ unordered pairs, with each bond appearing in $k-1$ of those pairs. Thus bond pairs in fact account for both kinds of distances, so that individual bonds need not be further considered. 

\begin{figure}[b]
    \centering
    \includegraphics[width=\columnwidth]{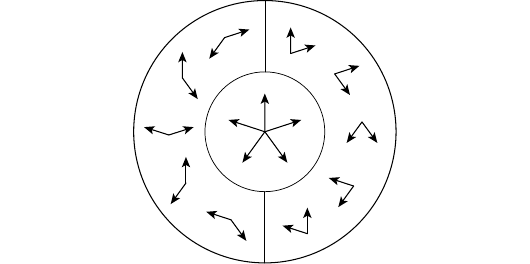}
    \caption{A regular pentagonal bond set (center) and the partition of its bond pairs by angle: $72^\circ$ (right) and $144^\circ$ (left).}
    \label{fig:pentagonal-bond-set}
\end{figure}

If all $k$ neighbors lie in the first coordination shell, as would be the case for a conventional neighbor relation, then any differences in bond lengths will be small relative to the distances between neighbors. Indeed such differences vanish altogether in the zero-temperature limit of many common crystal structures. In this way, pairs of bonds are characterized principally by the angles they make, illustrated for a two-dimensional bond set in Fig.~\ref{fig:pentagonal-bond-set}.

Let $n_i$ denote the number of bond pairs that make the angle $\vartheta_i$. By appeal to the original observation of Hartley \cite{Hartley1928} the information content of such a bond pair is given logarithmically by $\log_2 n_i$. We write this as follows:
\begin{equation*}
    H( \{b, b'\} \mid \angle\{b, b'\} = \vartheta_i ) = \log_2 n_i.
\end{equation*}

Following Shannon's subsequent theory \cite{Shannon1948, Khinchin1957} the information content of a generic bond pair $\{b,b'\}$ after accounting for its angle $\angle\{b,b'\}$ can be written as the weighted average of the above logarithm over all $m$ bond angle classes. The result is what one would formally describe as the conditional entropy of a bond pair given its angle,
\begin{align}
    H( \{b, b'\} \mid \angle\{b,b'\} ) &= \sum_{i=1}^m \frac{n_i}{n} H( \{b, b'\} \mid \angle\{b, b'\} = \vartheta_i ) \nonumber \\
    &= \frac{1}{n}\sum_{i=1}^m n_i \log_2 {n_i}
    \label{eq:parameter-definition}.
\end{align}
We term this entropy the \textit{extracopularity parameter} $\mathcal{E}$.

Given a bond pair, the corresponding bond angle is uniquely determined by a simple geometric calculation. But given a bond angle, the corresponding bond pair is determined only up to the finitely many pairs that make that same angle. The greater the angular redundancy among bond pairs, the more the residual information in a generic bond pair once its angle is observed. Since a bond pair is nothing more than a pair of bonds, each giving the relative position of a neighbor of the central particle, the extracopularity parameter is precisely a measure of the angular redundancy in neighboring particle positions.

\subsubsection{First results}
\label{subsubsec:first-results}

Eq.~\eqref{eq:parameter-definition} provides the information-theoretic definition of the extracopularity parameter $\mathcal{E}$. Given a tuple of bond angle multiplicities $(n_1, \dots, n_m)$, define the microscopic bond angle entropy $H(\theta)$ of a particle by
\begin{equation}
    H(\theta) = -\sum_{i = 1}^m \frac{n_i}{n} \log_2 \Bigl( \frac{n_i}{n} \Bigr), \quad n = \sum_{i=1}^m n_i.
    \label{eq:naive-bond-angle-entropy}
\end{equation}
Below we show that $\mathcal{E}$ can be expressed also in terms of the coordination number $k$ and the bond angle entropy $H(\theta)$.

\begin{lemma}
$ \mathcal{E} = \log_2\mathopen{}\tbinom{k}{2} - H(\theta).$
\label{lem:bayes}
\end{lemma}
\begin{proof}
We have from Eq.~\eqref{eq:parameter-definition} that
\begin{equation*}
    \mathcal{E} = H( \{b, b'\} \mid \angle\{b, b'\} ).
\end{equation*}
Expanding the conditional entropy gives
\begin{multline*}
    H( \{b, b'\} \mid \angle\{b, b'\} ) = H( \angle\{b, b'\} \mid \{b, b'\} )
     \\ + H(\{b, b'\}) - H(\angle\{b, b'\}).
\end{multline*}
We evaluate the terms on the right-hand side in their order of appearance. Since bond angles are uniquely determined by bond pairs,
\begin{equation*}
     H( \angle\{b, b'\} \mid \{b, b'\} ) = 0.
\end{equation*}
A generic bond pair $\{b, b'\}$ takes one of $n = \binom{k}{2}$ values. Its Shannon entropy is therefore given by
\begin{align*}
    H ( \{b, b'\} ) &= -\sum_{i=1}^m \frac{1}{n} \log_2 \Bigl(\frac{1}{n}\Bigr) \\
    &= \log_2 n \\
    &= \log_2\mathopen{}\tbinom{k}{2}.
\end{align*}
And the entropy of the angle between a generic pair of bonds is exactly the bond angle entropy:
\begin{align*}
    H(\angle\{b, b'\}) &= -\sum_{i = 1}^m \frac{n_i}{n} \log_2 \Bigl( \frac{n_i}{n} \Bigr) \\ &= H(\theta).
\end{align*}
From the above we obtain
\begin{equation*}
    \mathcal{E} = \log_2\mathopen{}\tbinom{k}{2} - H(\theta).
\end{equation*}
\end{proof}

Using Lemma~\ref{lem:bayes} it is easily shown that the extracopularity parameter satisfies the definition of a local order quantifier developed in Sec.~\ref{subsubsec:local-order-quantifier}. The proof of the following theorem uses the finiteness of the coordination number, itself a consequence of the finiteness of the number density.

\begin{theorem}
$\mathcal{E}$ is a local order quantifier.
\begin{proof}
By Lemma~\ref{lem:bayes},
\begin{equation*}
     \mathcal{E} = \log_2\mathopen{}\tbinom{k}{2} - H(\theta).
\end{equation*}
Property 1 of Definition~\ref{defn:local-order-quantifier} is then immediate given that both $k$ and $H(\theta)$ are functions of the bond set $B$. We verify that $\mathcal{E}$ possesses the remaining three properties:

\bigskip \noindent
{Property 2.}
\begin{equation*}
    H(\theta) \leq \log_2\mathopen{}\tbinom{k}{2} < \infty, \quad k < \infty.
\end{equation*}
{Property 3.}
\begin{equation*}
    \frac{\partial \mathcal{E}}{\partial k} = \frac{2k-1}{k(k-1)\ln 2} > 0, \quad k \geq 2.
\end{equation*}
{Property 4.}
\begin{equation*}
    \frac{\partial \mathcal{E}}{\partial H(\theta)} <0, \quad k \geq 2.
\end{equation*}
\end{proof}
\label{thm:local-order-quantifier}
\end{theorem}

Observe that the definition of $\mathcal{E}$ by Eq.~\eqref{eq:parameter-definition} leaves it undefined for $k \in \{0, 1\}$. Based on Lemma~\ref{lem:bayes}, property 2 of Definition~\ref{defn:local-order-quantifier}, and its information-theoretic interpretation, we extend the parameter as follows:
\begin{equation}
    \mathcal{E} = 
    \begin{cases}
        0 & k \in \{0, 1\} \\
        \log_2\mathopen{}\tbinom{k}{2}-H(\theta) & k \in \{2, 3, \dots\}.
    \end{cases}
\end{equation}

It is worthy of notice that $\mathcal{E}$ is in a certain sense the simplest function that satisfies Definition~\ref{defn:local-order-quantifier}. In particular, properties 3 and 4 are most easily attained by a function of additively separable form,
\begin{equation*}
    f(k, H(\theta)) = f_1(k) + f_2(H(\theta)),
\end{equation*}
with the unique homogeneous and isometric choice of $f_2$ that has property 4 being the negation function,
\begin{equation*}
    f_2(H(\theta)) = -H(\theta),
\end{equation*}
and the minimal choice of $f_1$ having property 3 being the least upper bound of the additive inverse of $f_2$,
\begin{equation*}
    f_1(k) = \log_2\mathopen{}\tbinom{k}{2},
\end{equation*}
which automatically ensures properties 1 and 2.

As a corollary to Theorem~\ref{thm:local-order-quantifier} we have that the extracopularity parameter $\mathcal{E}$ is bounded from below by an earlier descriptor of interest called the extracopularity \mbox{coefficient} $E$ \cite{Camkiran2022, Camkiran2023, Stimac2023, Hibbard2024}, given for a particle with coordination number $k \geq 2$ and $m \leq \binom{k}{2}$ bond angle classes by
\begin{equation}
    E = \log_2\biggl(\frac{k^2-k}{2m}\biggr).
    \label{eq:coefficient}
\end{equation}

\begin{corollary}
$\mathcal{E} \geq E.$
\begin{proof}
Since $H(\theta) \leq \log_2 m$ \cite[Theorem~2.6.4]{Cover2006},
    \begin{align*}
    \mathcal{E} &= \log_2\mathopen{}\tbinom{k}{2} - H(\theta) \nonumber \\[1ex]
    &\geq \log_2\mathopen{}\tbinom{k}{2} - \log_2 m \nonumber \\[1ex]
    &= \log_2 \biggl( \frac{k^2-k}{2m} \biggr) \nonumber \\
    &= E.
    \end{align*} 
\end{proof}
\label{corr:coefficient}
\end{corollary}

Table~\ref{tab:geometries-1} compares the parameter $\mathcal{E}$ with the coefficient $E$ for a few geometries of physical prevalence. Notice that the bound provided by Corollary~\ref{corr:coefficient} is attained in the regular tetrahedral case, indicating that it is sharp, and that the average percentage departure of $\mathcal{E}$ from $E$ is under ten, suggesting that it is tight also.

We examine the regular icosahedral case, which is notable for its high symmetry \cite{King1970, Coxeter1973}, low energy \cite{Frank1952, Cohn2007, Malins2013}, dense packing \cite{Hales2010, Schaller2016}, and small intertia \cite{Sloane1995, Manoharan2003}.

\renewcommand{\arraystretch}{1.1}
\begin{table}[b]
    \centering
    \begin{ruledtabular}
        \begin{tabular}{lrrrrr}
            \noalign{\vskip 2pt}
            Geometry
            & \multicolumn{1}{c}{$k$}
            & \multicolumn{1}{c}{$H(\theta)$}
            & \multicolumn{1}{c}{$\mathcal{E}$}
            & \multicolumn{1}{c}{$E$}
            & \multicolumn{1}{c}{$\Delta \%$}
            \\
            \noalign{\vskip 3pt}
            \hline
            \noalign{\vskip 4pt}
            Trigonal bipyramidal      &  $5$ & $1.296$ & $2.026$ & $1.737$ & $14.3$ \\
            Regular tetrahedral       &  $4$ & $0.000$ & $2.585$ & $2.585$ &  $0.0$ \\
            Pentagonal bipyramidal    &  $7$ & $1.705$ & $2.688$ & $2.392$ & $10.1$ \\
            Hexagonal bipyramidal     &  $8$ & $1.877$ & $2.930$ & $2.807$ &  $4.2$ \\
            Regular octahedral        &  $6$ & $0.722$ & $3.184$ & $2.907$ &  $8.7$ \\
            Square antiprismatic      &  $8$ & $1.379$ & $3.429$ & $3.222$ &  $6.0$ \\
            Triangular orthobicupolar & $12$ & $2.209$ & $3.835$ & $3.459$ &  $9.8$ \\
            Rhombic dodecahedral      & $14$ & $2.455$ & $4.053$ & $3.923$ &  $3.3$ \\
            Cuboctahedral             & $12$ & $1.823$ & $4.221$ & $4.044$ &  $4.2$ \\
            Regular icosahedral       & $12$ & $1.349$ & $4.696$ & $4.459$ &  $5.0$ \\
            \noalign{\vskip 1pt}
        \end{tabular}
    \end{ruledtabular}
    \caption{Coordination number $k$, bond angle entropy $H(\theta)$, extracopularity parameter $\mathcal{E}$, extracopularity coefficient $E$, and the percentage departure $\Delta \%$ of the parameter from the coefficient for ten commonly encountered geometries.}
    \label{tab:geometries-1}
\end{table}

\begin{example} 
Up to a scaling and an orthogonal transformation, a regular icosahedral bond set $B$ has the form
\begin{align*}
    B = \{(0, \pm 1, \pm \phi), (\pm 1, \pm \phi, 0), (\pm \phi, 0, \pm 1)\},
\end{align*}
where $\phi = (1+ \sqrt{5})/2$ is the golden ratio.
It will be seen that $B$ contains $k = 12$ bonds and therefore $n = 66$ bond pairs. Applying the formula
\begin{equation*}
    \angle \{b, b'\} = \arccos\biggl(\frac{b \cdot b}{\norm{b} \norm{b'}}\biggr)
\end{equation*}
yields the following multiplicities:
\begin{equation*}
    n_i =
    \begin{cases}
        30 & \theta_i = 63.4^\circ \\
        30 & \theta_i = 116.6^\circ \\
        \phantom{1}6 & \theta_i = 180^\circ.
    \end{cases}
\end{equation*}
Evaluating Eq.~\eqref{eq:naive-bond-angle-entropy} gives
\begin{align*}
    H(\theta) &= -\sum_{i = 1}^m \frac{n_i}{n} \log_2 \biggl(\frac{n_i}{n}\biggr) \\ 
    &=  -2 \biggl(\frac{30}{66}\biggr)\log_2\biggl(\frac{30}{66}\biggr) -\frac{6}{66}\log_2\biggl(\frac{6}{66}\biggr) \\
    &= \log_2 11 - \frac{10}{11}\log_2 5 \\[1ex]
    &\approx  1.349~\text{bits}.
\end{align*}
Using Lemma~\ref{lem:bayes} we obtain
\begin{align*}
    \mathcal{E} &= \log_2\mathopen{}\tbinom{k}{2} - H(\theta) \\
    &= \log_2 66 - \log_2 11 + \frac{10}{11}\log_2 5 \\
    &= \log_2 6 + \frac{10}{11}\log_2 5 \\[1ex]
    &\approx 4.696~\text{bits}.
\end{align*}
\label{example:ico}
\end{example}

\subsubsection{Bond angle entropy}
\label{subsubsec:bond-angle-entropy}

All of the geometries so far considered are characterized by well-separated bond angles. For such geometries, the microscopic bond angle entropy $H(\theta)$ can be computed directly by inserting the associated multiplicities $n_i$ into Eq.~\eqref{eq:naive-bond-angle-entropy} as in Example~\ref{example:ico}. But in the generic case where one has an empirical measure of the form
\begin{equation}
    p_\theta(\vartheta) = \frac{1}{n}\sum_{i = 1}^n \delta(\vartheta - \vartheta_i),
    \label{eq:empirical-measure}
\end{equation}
where $\delta$ is the Dirac delta function, bond angles may not be well separated. In such situations, the multiplicities are degenerate and the usual formula cannot be used.

The naive definition of $H(\theta)$ provided by Eq.~\eqref{eq:naive-bond-angle-entropy} can be extended to the generic case by interpolating the relationship between the (microscopic) bond angle entropy, where known, and the entropy of an embedded image of the empirical measure, which we now describe.

Before proceeding we observe that a valid assignment of bond angle entropy must come from the set of all values that can be produced by Eq.~\eqref{eq:naive-bond-angle-entropy} for the given $k \geq 2$,
\begin{multline}
    \mathbb{H}_k = \biggl\{ -\sum_{i=1}^m \frac{n_i}{n} \log_2 \Bigl( \frac{n_i}{n} \Bigr) : n = \tbinom{k}{2}, \\
    m \in \{1, \dots, n\}, \, n_i \in \mathbb{N}, \, \sum_{i=1}^m n_i = n \biggr\}.
    \label{eq:H-set}
\end{multline}
Let us write $\lfloor x \rceil_{\mathbb{H}_k}$ for the element of $\mathbb{H}_k$ nearest to $x \in \mathbb{R}$.

We define an embedding that convolves the empirical bond angle measure $p_\theta$ with a Gaussian kernel $g_\sigma$. This produces a density on $(-\infty, +\infty)$ given by
\begin{align*}
    (g_\sigma * p_\theta)(\vartheta)
    &= \frac{1}{n}\sum_{i = 1}^n g_\sigma(\vartheta - \vartheta_i) \nonumber \\
    &= \frac{1}{n} \sum_{i=1}^n \frac{1}{\sqrt{2\pi}\sigma}
    \exp\biggl(-\frac{(\vartheta - \vartheta_i)^2}{2\sigma^2}\biggr).
\end{align*}
The embedded image entropy $\mathfrak{H}(\theta)$ is then obtained by evaluating its differential Shannon entropy:
\begin{align}
    \mathfrak{H}(\theta) &=- \int_{-\infty}^{+\infty} (g_\sigma * p_\theta)(\vartheta) \ln [(g_\sigma * p_\theta)(\vartheta)] d\vartheta.
\end{align}

Consider the set $\mathbb{B}$ of all bond sets with well-separated bond angles. It can be shown using information theory that $H(\theta)$ is approximately affine in $\mathfrak{H}(\theta)$ for such bond sets: Given a discrete random variable $X$ with mass function $p_X$ define $H(p_X) = H(X)$. The Shannon entropy of the $\Delta$-quantized embedded image, $H((g_\sigma * p_\theta)_\Delta)$, satisfies
\begin{align*}
    H((g_\sigma * p_\theta)_\Delta) &= \frac{\mathfrak{H}(\theta)}{\ln 2} + \log_2 \frac{1}{\Delta} + r_\Delta(\Delta, \sigma) \\
    &= H((g_\sigma)_\Delta) + H(\theta) + r_\sigma(\Delta, \sigma),
\end{align*}
with $r_\Delta(\Delta, \sigma) \to 0$ as $\Delta \to 0$ by the entropy of a quantized continuous random variable \cite[Theorem~8.3.1]{Cover2006} and $r_\sigma(\Delta, \sigma) \to 0$ as $\sigma \to 0$ under well-separated bond angles by the entropy of a disjoint mixture \cite[Lemma~2.5.2]{Gray2011}. Let $\epsilon(\Delta, \sigma) = r_\Delta(\Delta, \sigma) - r_\sigma(\Delta, \sigma)$. Then
\begin{align*}
    H(\theta) &= \frac{1}{\ln 2} \mathfrak{H}(\theta) + \log_2 \frac{1}{\Delta} - H((g_\sigma)_\Delta) + \epsilon(\Delta, \sigma) \\
    &\approx c_1 \mathfrak{H}(\theta) + c_0(\Delta, \sigma), \enspace c_1 = 1/\ln 2, \enspace c_0(\Delta, \sigma) \in \mathbb{R}.
\end{align*}
Verifying this result via regression, taking $\mathbb{B}$ to be the set of bond sets in Table~\ref{tab:geometries-1}, gives an $R^2$ within $10^{-4}$ of one.

Notwithstanding this approximate affinity, a true extension of Eq.~\eqref{eq:naive-bond-angle-entropy} must recover the exact value of $H(\theta)$ for every element of $\mathbb{B}$. Given a bond set $B_i \in \mathbb{B}$ denote by $H(\theta_i)$ its bond angle entropy and by $\mathfrak{H}(\theta_i)$ the entropy of the associated embedded image. We satisfy the above requirement with the following interpolation:
\begin{equation}   
    \begin{gathered}
        H(\theta) = \lfloor(1-t) H(\theta_a) + t H(\theta_b)\rceil_{\mathbb{H}_k}, \\[1ex]
        t = \frac{\mathfrak{H}(\theta) - \mathfrak{H}(\theta_a)}{\mathfrak{H}(\theta_b) - \mathfrak{H}(\theta_a)}; 
    \end{gathered}
    \label{eq:extended-bond-angle-entropy}
\end{equation}
where if $\mathfrak{H}(\theta_i) \leq \mathfrak{H}(\theta) < \mathfrak{H}(\theta_j)$ for some ${B_i, B_j \in \mathbb{B}}$, then $B_a$ and $B_b$ are two bond sets in $\mathbb{B}$ with, respectively, the largest and smallest embedded image entropies satisfying
\begin{equation*}
    \mathfrak{H}(\theta_a) \leq \mathfrak{H}(\theta) < \mathfrak{H}(\theta_b);
\end{equation*}
otherwise $B_a$ and $B_b$ are two bond sets in $\mathbb{B}$ that minimize and maximize the embedded image entropy, respectively.

The compromise between a degenerate embedding at ${\sigma=0^\circ}$ and the loss of bond angle class distinctions for $\sigma$ sufficiently large suggests that the optimal embedding scale is half the value of $\sigma$ for which the two closest peaks of $g_\sigma * p_\theta$ become indistinguishable for any bond set in $\mathbb{B}$. For the geometries in Table~\ref{tab:geometries-1}, this value is smallest in the triangular orthobicupolar case at $\sigma \approx 2.5^\circ$.

\renewcommand{\arraystretch}{1.1}
\begin{table}[b]
    \centering
    \begin{ruledtabular}
        \begin{tabular}{lrrr}
            \noalign{\vskip 2pt}
            Geometry
            & \multicolumn{1}{c}{$k$}
            & \multicolumn{1}{c}{$H(\theta)$}
            & \multicolumn{1}{c}{$\mathcal{E}$}
            \\
            \noalign{\vskip 3pt}
            \hline
            \noalign{\vskip 4pt}
            Trigonal antiprismatic           &  $6$ & $1.519$ & $2.388$ \\
            Capped trigonal prismatic        &  $7$ & $1.590$ & $2.802$ \\
            Snub disphenoidal                &  $8$ & $1.979$ & $2.829$ \\
            Bicapped trigonal prismatic      &  $8$ & $1.838$ & $2.969$ \\
            Tricapped trigonal prismatic     &  $9$ & $1.949$ & $3.221$ \\
            Capped square antiprismatic      &  $9$ & $1.865$ & $3.305$ \\
            Bicapped square antiprismatic    & $10$ & $2.085$ & $3.407$ \\
            Bicapped pentagonal prismatic    & $12$ & $2.260$ & $3.785$ \\
            Octadecahedral                   & $11$ & $1.942$ & $3.840$ \\
            Capped pentagonal antiprismatic  & $11$ & $1.349$ & $4.433$ \\
            \noalign{\vskip 1pt}
        \end{tabular}
    \end{ruledtabular}
    \caption{Coordination number $k$, bond angle entropy $H(\theta)$, and extracopularity parameter $\mathcal{E}$ for ten other geometries.}
    \label{tab:geometries-2}
\end{table}

Eq.~\eqref{eq:extended-bond-angle-entropy} has both of the two properties that would be expected of an extension of the microscopic bond angle entropy, namely that it agrees with the naive definition provided by Eq.~\eqref{eq:naive-bond-angle-entropy} for all $B \in \mathbb{B}$ and that it is continuous with respect to the empirical measure $p_\theta$ [Eq.~\eqref{eq:empirical-measure}]. Values of $\mathcal{E}$ for geometries characterized by ill-separated (or otherwise ambiguous) bond angles, with $H(\theta)$ computed using the above extension, are given in Table~\ref{tab:geometries-2}.

\subsection{Order and symmetry}
\label{subsec:order-symmetry}

The study of the structure of many-particle systems is often concerned with symmetry, which in addition to its significance to crystals, is also important in liquids and glasses \cite{Finney1989, Leocmach2012, Tanaka2013, Malins2013, Hu2015, Zhang2020}. Yet the link between order and symmetry has hitherto received little attention \cite{Pinsky1998, Gavezzotti2007, Tanaka2019, Hibbard2024}, with much of the related effort directed toward quantifying how close a local cluster is to having a certain kind of symmetry \cite{Nelson1979, Steinhardt1983, Pinsky1998, Kelchner1998, Tsuzuki2007, Cumby2017, Han2020, Liu2024}. Below we establish a precise relationship between these formally distinct aspects with the help of extracopularity.

\subsubsection{Point group bound}
\label{subsubsec:point-group-bound}

Recall from Sec.~\ref{subsubsec:local-order-quantifier} that a bond set $B$ is a collection of $k$ vectors or bonds in $D$ dimensions, where $k$ is called the coordination number. It is helpful to introduce an idealization in which all bonds in $B$ are equal in length.

\begin{definition}[spherical]
A bond set $B$ is said to be \textit{spherical} if it has $k > 2$ elements of the same length.
\end{definition}

Other things being equal, a larger point group suggests a further departure from complete randomness and thus a greater degree of order. The following theorem asserts, in accord with this intuition, that the extracopularity parameter for a spherical bond set is bounded from below by an increasing function of the size of its point group.

\begin{theorem}
Let $B$ be a spherical bond set with point group $G$. Denote by $\sigma_O$ the size of the stabilizer of a bond pair with orbit $O$ and define $\sigma^* = \max_{O}\sigma_O$. Then
\begin{equation*}
    \mathcal{E} \geq E \geq \log_2 \biggl( \frac{\abs{G}}{\sigma^*} \biggr).
\end{equation*}
\label{thm:symmetry}
\begin{proof}
Consider the set of all pairs of bonds in $B$,
\begin{equation*}
    \tbinom{B}{2} = \{\{b,b'\} \subset B: b \neq b' \}.
\end{equation*}
The action of $G$ on this set is defined by
\begin{equation*}
    g\{b,b'\} = \{gb,gb'\}
\end{equation*}
and induces a homomorphism
\begin{equation*}
    \Phi : G \to \operatorname{Sym}\bigl( \tbinom{B}{2}\bigr).
\end{equation*}
With $k>2$, every nontrivial element of $G$ moves at least one bond pair. The action is therefore faithful, so that
\begin{equation*}
\abs{\Phi(G)} = \abs{G}.
\end{equation*}
By the orbit–stabilizer theorem \cite{Gallian2021} the size of the orbit $O$ of a bond pair with stabilizer $\sigma_O$ is then given by
\begin{align}
    \abs{O} 
    &= \frac{\abs{\Phi(G)}}{\abs{\sigma_O}} \nonumber \\
    &= \frac{\abs{G}}{\abs{\sigma_O}}.
\end{align}

Since the set $\mathcal{O}$ of all $\Phi(G)$-orbits partitions $\tbinom{B}{2}$, the sum of the sizes of all such orbits is equal to the total bond pair count
\begin{align}
    n 
    &= \bigl\lvert{\tbinom{B}{2}}\bigr\rvert \nonumber \\[2ex]
    &= \sum_{O \in \mathcal{O}} \abs{O} \nonumber \\ 
    &= \abs{G} \sum_{O \in \mathcal{O}} \frac{1}{\abs{\sigma_O}}. 
    \label{eq:sum-orbit-sizes}
\end{align}

The action of $G$ preserves angles, and so the number $m$ of angle classes cannot exceed the number of orbits:
\begin{equation}
    m \leq \abs{\mathcal{O}}.
    \label{eq:m-bound}
\end{equation}

Dividing both sides of Eq.~\eqref{eq:sum-orbit-sizes} by $m$ and then using Ineq.~\eqref{eq:m-bound} gives
\begin{align}
    \frac{n}{m} 
    &= \frac{\abs{G}}{m} \sum_{O \in \mathcal{O}} \frac{1}{\sigma_O} \nonumber \\ 
    &\geq \abs{G} \frac{1}{\abs{\mathcal{O}}} \sum_{O \in \mathcal{O}} \frac{1}{\sigma_O} \nonumber \\
    &\geq \frac{\abs{G}}{\sigma^*},
\end{align}
where the final inequality follows from the fact that $1/\sigma_O \geq 1/\sigma^*$ for all $O \in \mathcal{O}$. Now recalling Corollary~\ref{corr:coefficient}  we arrive at
\begin{align*}
    \mathcal{E} 
    &\geq E \\
    &= \log_2\Bigl(\frac{n}{m}\Bigr) \\
    &\geq \log_2\biggl(\frac{\abs{G}}{\sigma^*}\biggr).
\end{align*}
\end{proof}
\end{theorem}
 
The second inequality in the statement of Theorem~\ref{thm:symmetry} can be shown to hold with equality for bond sets that have the geometry of a regular simplex or that of a regular convex polygon with an odd number of vertices. The following example demonstrates the simplicial case in three dimensions, where the bond set has a regular tetrahedral geometry.

\begin{example}
Let $B$ be a regular tetrahedral bond set. We have $k = 4$ bonds, equal in length; $m = 1$ bond angle class; and $G=T_d$ as the point group, so that $\abs{G} = 24$. The corresponding set of bond pairs takes the form
\begin{align*}
\tbinom{B}{2} = \{ &\{b_1,b_2\}, \{b_1,b_3\}, \{b_1,b_4\}, \\
                             &\{b_2,b_3\}, \{b_2,b_4\}, \{b_3,b_4\} \}.
\end{align*}

Consider the stabilizer of the pair $\{b_1,b_2\}$ under the action of $G$. As illustrated in Fig.~\ref{fig:bond-pair-stabilizer}, this set has four elements: an identity
\begin{equation*}
e : 
    \begin{cases}
        \{b_1, b_2\} \mapsto \{b_1, b_2\} \\
        \{b_1, b_3\} \mapsto \{b_1, b_3\} \\
        \{b_1, b_4\} \mapsto \{b_1, b_4\} \\
        \{b_2, b_3\} \mapsto \{b_2, b_3\} \\
        \{b_2, b_4\} \mapsto \{b_2, b_4\} \\
        \{b_3, b_4\} \mapsto \{b_3, b_4\}
    \end{cases}
\end{equation*}
a rotation
\begin{equation*}
R : 
    \begin{cases}
        \{b_1, b_2\} \mapsto \{b_2, b_1\} \\
        \{b_1, b_3\} \mapsto \{b_2, b_4\} \\
        \{b_1, b_4\} \mapsto \{b_2, b_3\} \\
        \{b_2, b_3\} \mapsto \{b_1, b_4\} \\
        \{b_2, b_4\} \mapsto \{b_1, b_3\} \\
        \{b_3, b_4\} \mapsto \{b_4, b_3\}
    \end{cases}
\end{equation*}
and two reflections
\begin{equation*}
M : 
    \begin{cases}
        \{b_1, b_2\} \mapsto \{b_1, b_2\} \\
        \{b_1, b_3\} \mapsto \{b_1, b_4\} \\
        \{b_1, b_4\} \mapsto \{b_1, b_3\} \\
        \{b_2, b_3\} \mapsto \{b_2, b_4\} \\
        \{b_2, b_4\} \mapsto \{b_2, b_3\} \\
        \{b_3, b_4\} \mapsto \{b_4, b_3\}
    \end{cases}
 M' : 
    \begin{cases}
        \{b_1, b_2\} \mapsto \{b_2, b_1\} \\
        \{b_1, b_3\} \mapsto \{b_2, b_3\} \\
        \{b_1, b_4\} \mapsto \{b_2, b_4\} \\
        \{b_2, b_3\} \mapsto \{b_1, b_3\} \\
        \{b_2, b_4\} \mapsto \{b_1, b_4\} \\
        \{b_3, b_4\} \mapsto \{b_3, b_4\}.
    \end{cases}
\end{equation*}
\begin{figure}[t]
    \centering
    \includegraphics[width=\columnwidth]{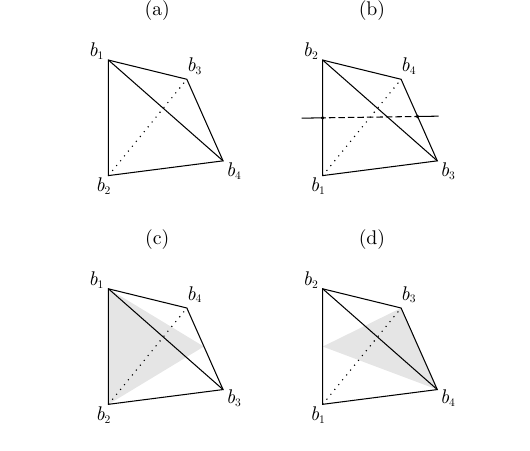}
    \caption{The stabilizer of the bond pair $\{b_1,b_2\}$ under the action of the tetrahedral point group $T_d$: (a) the identity $e$, (b) the rotation $R$, (c) the reflection $M$, and (b) the reflection $M'$. Here bonds are represented not by arrows as is usually done but by the vertices at their terminal points, so that bond pairs correspond to the line segments joining those vertices.} 
    \label{fig:bond-pair-stabilizer}
\end{figure}
By the edge-transitivity of the regular tetrahedron, all $n=6$ pairs in $\binom{B}{2}$ have the same orbit and therefore stabilizers of the same size $\sigma^* = \max\{4\} = 4$. Thus
\begin{align*}
    \mathcal{E} &= E \\ 
    &= \log_2\Bigl(\frac{n}{m}\Bigr) \\ 
    &= \log_2\Bigl(\frac{4n}{4m}\Bigr) \\
    &= \log_2\biggl(\frac{\abs{G}}{\sigma^*}\biggr).
\end{align*}
\end{example}

\subsubsection{Icosahedral maximum}
\label{subsubsec:icosahedral-maximum}

The regular icosahedron has the largest point group among all convex polyhedra with a vertex count of at most $12$, this being the kissing number in three dimensions \cite{Schuette1952, Musin2006}. We now prove an analogous result for the lower bound $E$ of the extracopularity parameter.

\begin{proposition}
Let $B$ be a spherical bond set in $\mathbb{R}^3$ with coordination number $k = \abs{B} \leq 12$. Then $E \leq E_\text{ico}$.
\begin{proof}
Let $E = E(k, m)$. A regular icosahedral bond set places $k=12$ points on a sphere with $m=3$ distinct central angles. Using Eq.~\eqref{eq:coefficient} we have $E_\text{ico} = E(12,3) = \log_2(22)$. We show that no valid combination of $k$ and $m$ gives higher $E$. Since $E$ is strictly increasing in $k$, it suffices to check the largest admissible $k$ for each $m$. And since $E$ is strictly decreasing in $m$, we need not check $m > 3$. For $m = 3$, we have $E(k,3) \leq E(12,3)$ by the hypothesis that $k \leq 12$. For $m = 2$, no more than $k=6$ points can be placed on a sphere \cite[Theorem~2]{Croft1962}; hence $E(k,2) \leq E(6, 2) = \log_2(15/2) < \log_2(22) = E(12,3)$. For $m = 1$, no more than $k = 4$ points can be placed on a sphere; hence $E(k,1) \leq E(4,1) = \log_2(6) < \log_2(22) = E(12,3)$. Thus $E \leq E(12 ,3) = E_\text{ico}$.
\end{proof}
\label{prop:icosahedral}
\end{proposition}

\subsection{Distribution function}
\label{subsec:distribution-function}

By Theorem~\ref{thm:local-order-quantifier}, the extracopularity parameter $\mathcal{E}$ is a local order quantifier in the sense of Definition~\ref{defn:local-order-quantifier} and hence a source of coarse local descriptions as defined in Sec.~\ref{subsubsec:coarse-local-descriptions}. Accordingly some function of the values of $\mathcal{E}$ for all particles in a system approximates a complete invariant for its structure. In many cases, however, the statistical behavior of a local descriptor for a single particle suffices to draw distinctions between structures. We therefore focus the subsequent discussion on the marginal distribution of $\mathcal{E}$.

Define the \textit{\mbox{extracopularity} distribution function} (XDF) $p_\mathcal{E}$ as the marginal probability mass function of the extracopularity parameter for a given particle. We may write the conditional distribution of $\mathcal{E}$ given $k$ piecewise.
\begin{multline*}
    P\{ \mathcal{E} = \varepsilon \mid k = \kappa\} \\
    = 
    \begin{cases}
        \boldsymbol{1}\{ \varepsilon = 0 \} & \kappa \in \{0, 1, 2\}  \\[.5ex]
        P \bigl\{ \log_2\mathopen{}\tbinom{k}{2} - H(\theta) = \varepsilon \bigm| k = \kappa\bigr\} & \kappa \in \{3, 4, \dots\},
    \end{cases}
\end{multline*}
where $\boldsymbol{1}\{A\}$ denotes the indicator function of the condition $A$. The XDF can then be expressed as a mixture of the following form:
\begin{align}
    p_\mathcal{E}(\varepsilon) &= P\{\mathcal{E} = \varepsilon\} \nonumber \\
    &= \sum_{\kappa = 0}^\infty p_k(\kappa) P\{ \mathcal{E} = \varepsilon \mid k = \kappa\} \nonumber \\
    &= (p_k(0) + p_k(1) + p_k(2))\boldsymbol{1}\{\varepsilon = 0\} \nonumber \\ 
    &\hphantom{=}\,\,+\sum_{\kappa = 3}^ \infty p_k(\kappa) P\bigl\{ H(\theta) = \log_2\mathopen{}\tbinom{k}{2} - \varepsilon \bigm| k =\kappa\bigr\}.
    \label{eq:xdf-master}
\end{align}
The second factor of each term in the series is determined by the conditional distribution of the microscopic bond angle entropy $H(\theta)$ given the coordination number $k$. Below this function is approximated for particles with high coordination numbers and broad bond angle distributions.

Recall from Sec.~\ref{subsubsec:bond-angle-entropy} that the (microscopic) bond angle entropy $H(\theta)$ is a functional of the empirical bond angle measure $p_\theta$. This measure is itself a function of a sample from the ensemble bond angle distribution $f_\Theta$, given for a tuple of ${\nu = \tbinom{\kappa}{2}}$ angles $(\Theta^{(1)}, \dots, \Theta^{(\nu)})$ by

\begin{multline}
    f_\Theta(\vartheta^{(1)}, \dots, \vartheta^{(\nu)}) \\[0.5ex] = \frac{1}{\nu!} \Biggl\langle \, \sum_{\pi \in \operatorname{Sym}(\{1, \dots, \nu\})} \prod_{j = 1}^\nu \delta(\Theta^{\pi(j)} - \vartheta^{(j)}) \Biggr\rangle.
    \label{eq:joint-ensemble-bond-angle-distrib}
\end{multline}
Hence the bond angle entropy as it is formally extended in Eq.~\eqref{eq:extended-bond-angle-entropy} is better suited to numerical evaluation.

In order to proceed symbolically, one may work with a tractable estimator of $H(\theta)$. We start by partitioning the bond angle interval $(0^\circ, 180^\circ]$ into $M$ equal-width bins
\begin{equation*}
    I_i =  \biggl(\frac{180^\circ}{M}(i-1), \frac{180^\circ}{M}i\biggr], \quad i = 1, \dots, M.
\end{equation*}
Let $\hat{n}_i$ denote the number of bond pairs whose angles lie in the $i$th bin, so that $\sum_{i=1}^M \hat{n}_i = \nu$. Then $(\hat{n}_1, \dots, \hat{n}_M)$ is an $M$-bin bond angle histogram and the corresponding maximum likelihood estimator of the bond angle entropy takes the form
\begin{equation}
    \hat{H}(\theta) = - \sum_{i = 1}^M \frac{\hat{n}_i}{\nu} \log_2 \biggl(\frac{\hat{n}_i}{\nu}\biggr). 
    \label{eq:max-likelihood}
\end{equation}

We suppose that each histogram is constructed from $\nu$ independent samples from the discretized one-angle marginal $p_\Theta$, defined for the $M$ bin centers $\vartheta_i$ by
\begin{align}
    p_\Theta(\vartheta_i) =  \int_{(0^\circ, 180^\circ]^{\nu-1}} \int_{I_i}  f_\Theta(\vartheta',\vartheta^{(2)}, \dots, \vartheta^{(\nu)}) d \vartheta' d^{\nu-1}\vartheta.
    \label{eq:disc-ensemble-bond-angle-distrib}
\end{align}
Under this simplifying assumption we have multinomial histogram probabilities
\begin{multline*}
    \!\!\! P\{\hat{n}_1 = \nu_1,  \dots, \hat{n}_M = \nu_M\} = \frac{\nu!}{\nu_1!  \dots \nu_M !} \prod_{i = 1}^M p_\Theta(\vartheta_i)^{\nu_i}.
\end{multline*}
The conditional probabilities of the bond angle entropy estimates are then given by
\begin{align*}
    P\{\hat{H}(\theta) = \mathcal{H} \mid k = \kappa \} &= P\{\hat{H}(\theta) = \mathcal{H} \mid n = \nu \} \nonumber 
    \\ &= \sum_{\substack{(\nu_1, \dots, \nu_M): \\ \hat{H}(\theta) = 
    \mathcal{H}}} \nu! \prod_{i=1}^M \frac{p_\Theta(\vartheta_i)^{\nu_i}}{\nu_i!}.
    \label{eq:multinomial}
\end{align*}

For large $\nu$ and comparable $p_\Theta(\theta_i)$ the foregoing may be approximated by a normal distribution
\begin{equation}
    \mathcal{N}[\mu_{\hat{H}(\theta)|k}(\kappa), \sigma^2_{\hat{H}(\theta)|k}(\kappa)]
    \label{eq:normal-approx}
\end{equation}
with mean $\mu_{\hat{H}(\theta)|k}(\kappa)$ and variance $\sigma^2_{\hat{H}(\theta)|k}(\kappa)$.

The mean is obtained by subtracting the bias term from the first-order expectation expansion \cite{Miller1955, Paninski2003}:
\begin{equation}
    \mu_{\hat{H}(\theta)|k}(\kappa)\approx - \sum_{i=1}^M p_\Theta(\vartheta_i) \log_2  p_\Theta(\vartheta_i) - \frac{m-1}{2\binom{\kappa}{2}\ln2},
    \label{eq:mean-H}
\end{equation}
where $m \leq M$ is the number of nonempty bins and $0 \log_2 0 = 0$ by convention.

The variance is in turn obtained using the delta method \cite{Vaart1998}, which linearizes the entropy functional via a first-order Taylor expansion:
\begin{multline}
    \sigma_{\hat{H}(\theta)|k}^2(\kappa) \approx \frac{1}{\binom{\kappa}{2}} \biggl\{ \sum_{i=1}^M p_\Theta(\vartheta_i) \biggl( \log_2 p_\Theta(\vartheta_i) + \frac{1}{\ln2} \biggr)^2 \\ 
    - \biggl[\,\sum_{i=1}^M p_\Theta(\vartheta_i) \biggl( \log_2 p_\Theta(\vartheta_i) + \frac{1}{\ln2} \biggr) \biggr]^2 \biggr\}.
    \label{eq:variance-H}
\end{multline}

Together Eqs.~\eqref{eq:xdf-master} and \eqref{eq:max-likelihood}--\eqref{eq:normal-approx} produce a formula for the XDF of a highly coordinated particle with broadly distributed bond angles:
\begin{multline}
    Z_\mathcal{E}p_\mathcal{E}(\varepsilon) \approx \boldsymbol{1}\{\varepsilon = 0\}\sum_{\kappa = 0}^2 p_k(\kappa)  + \sum_{\kappa=3}^K {p_k(\kappa)} \\ 
    \frac{1}{\sigma_{\hat{H}(\theta)|k}(\kappa)}\exp\Biggl\{{-\frac{1}{2} \Biggl(\frac{\log_2\mathopen{}\tbinom{\kappa}{2} - \varepsilon-\mu_{\hat{H}(\theta)|k}(\kappa)}{\sigma_{\hat{H}(\theta)|k}(\kappa)}\Biggr)^2}\Biggr\}, \\[2ex] 
    \varepsilon \in \bigl\{ \log_2\mathopen{}\tbinom{a}{2} - b : a \in \mathbb{N}, 2 \leq a \leq K,  b \in \mathbb{H}_a \bigr\};
    \label{eq:xdf-highly-coordinated}
\end{multline}
with the series in Eq.~\eqref{eq:xdf-master} here truncated at some $K \gg 3$.

Eq.~\eqref{eq:xdf-highly-coordinated} is in essence a Gaussian mixture with mixture weights $p_k(\kappa)$, component means
\begin{equation*}
    \mu_{\mathcal{E}|k}(\kappa) = 
    \begin{cases}
        0 & \kappa \in \{0, 1\} \\
        \log_2\mathopen{}\tbinom{\kappa}{2} - \mu_{\hat{H}(\theta)|k}(\kappa) & \kappa \in \{2, 3, \dots\},
    \end{cases}
\end{equation*}
and component variances
\begin{equation*}
    \sigma^2_{\mathcal{E}|k}(\kappa) = \sigma^2_{\hat{H}(\theta)|k}(\kappa), \quad \kappa \in \{2, 3, \dots\}.
\end{equation*}
Using the laws of total expectation and variance its first two moments reduce respectively to
\begin{gather*}
    \mu_\mathcal{E} = \sum_{\kappa=3}^K p_k(\kappa) \bigl[ \log_2\mathopen{}\tbinom{\kappa}{2} -  \mu_{\hat{H}(\theta)|k}(\kappa)\bigr], \\
    \begin{split}
    \sigma_{\mathcal{E}}^{2} = \sum_{\kappa=3}^{K} &p_{k}(\kappa)\sigma_{\hat{H}(\theta)|k}^{2}(\kappa) + \sum_{\kappa=0}^{2}p_k(\kappa)(0-\mu_{\mathcal{E}})^2 \\ &+\sum_{\kappa=3}^{K}p_{k}(\kappa)\bigl[\log_{2}\tbinom{\kappa}{2} -\mu_{\hat{H}(\theta)|k}(\kappa) -\mu_{\mathcal{E}}\bigr]^{2}.
    \end{split}
\end{gather*}

It may be noted that the first summation in the expression for variance corresponds to the within-component and the sum of the second two to the between-component contributions to the total variance. Higher moments of the XDF can be computed in like manner.

\section{Elementary materials}
\label{sec:elementary-materials}

Certain of the many possible forms of matter occupy a place of special importance in our theoretical understanding of structure in many-particle systems. In what follows we use the distribution function introduced above to study the structure of the most elementary manifestations of the three conventional states of matter: the ideal gas, the perfect crystal, and the simple liquid.

\subsection{Ideal gas}
\label{subsec:ideal-gas}

The usual starting point for discussing structure in particle systems is the bulk ideal gas, wherein the absence of interactions and the negligibility of boundaries result in a trivial behavior that nonetheless proves instructive.

Let us begin by considering the distribution of the angle between the neighbor vectors of an ideal gas particle. Vectors of this kind are analogous to bonds in systems with interactions and will therefore be referred to as such.

Absent interactions and boundaries, individual bonds exhibit no preference for any particular direction in space. However, as seen in Fig.~\ref{fig:ideal-gas-angles}(a), fixing one bond $b$ renders the probability of observing the second bond $b'$ at an angle $\Theta \leq \vartheta$ from $b$ proportional to the spherical cap surface area
\[
    A_\vartheta = 2\pi r^2 (1-\cos \vartheta).
\]
Normalizing $A_\vartheta$ by the total area of the sphere gives the cumulative probability 
\begin{align*}
    P\{ \Theta \leq \vartheta \} \nonumber &= \frac{A_\vartheta}{4\pi r^2} \\ 
                                           &= \frac{1-\cos 2\vartheta}{2}.
\end{align*}
Differentiating the above expression yields the ensemble bond angle distribution for the bulk ideal gas:
\begin{align}
    f_\Theta(\vartheta) &= \frac{d}{d\vartheta} P\{\Theta \leq \vartheta\} \nonumber \\
    &= \frac{1}{2} \sin \vartheta, \quad 0^\circ < \vartheta \leq 180^\circ.
    \label{eq:ideal-gas-entropy}
\end{align}
Fig.~\ref{fig:ideal-gas-angles}(b) illustrates this distribution.

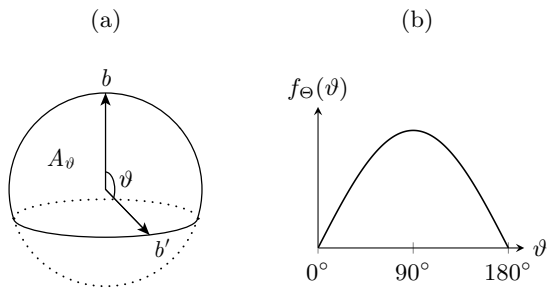
\begin{figure}[b]
    \centering
    \begin{tabular}{cc}
        (a) & (b) \\ \\
        \includegraphics[width=0.5\columnwidth]{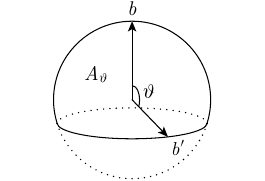} &
        \begin{tikzpicture}
            \begin{axis}[
                width=0.5\columnwidth,
                height=0.4\columnwidth,
                xmin=0,
                xmax=195,
                ymin=0,
                ymax=1.2,
                xtick={90, 180},
                xticklabels={$90^\circ$, $180^\circ$},
                extra x ticks={0},
                extra x tick labels={$0^\circ$},
                ytick=\empty,
                xlabel={$\vartheta$},
                ylabel={$f_\Theta(\vartheta)$},
                samples=720,
                axis lines=middle,
                axis line style={->,>=stealth},
                xlabel style={at={(axis description cs:1,0.00)}, anchor=west},
                ylabel style={at={(axis description cs:0,0.98)}, anchor=south}
            ]
                \addplot[
                    black,
                    semithick,
                    domain=0:180
                ]
                {sin(x)};
            \end{axis}
        \end{tikzpicture}
    \end{tabular}
    \caption{(a) In an ideal gas, the probability of two notional bonds $b$ and $b'$ making an angle $\Theta \leq \vartheta$ is proportional to the spherical cap area $A_\vartheta$. (b) The ensemble bond angle distribution of an ideal gas particle is given by a sine function.}
    \label{fig:ideal-gas-angles}
\end{figure}

We now turn to the coordination number $k$. Recall that an ideal gas of density $\rho$ can be regarded as the point-particle limit of a hard-sphere fluid with that same density. Indeed as the particle radius $r_0$ approaches zero the hard-sphere fluid converges to a homogeneous Poisson point process \cite{Chiu2013}, for which the probability of finding $N$ particles in a region of volume $V$ is equal to
\begin{equation*}
    \frac{(\rho V)^N}{N!} \exp( -\rho V).
\end{equation*}

It is an empirical observation that the location $r_1$ of the first RDF minimum and the particle radius $r_0$ satisfy the inequality $r_1 \leq 4r_0$ across packing fractions in three-dimensional hard-sphere fluids \cite{Rintoul1996, Royall2007, Pieprzyk2019}. Let $K$ be the number of particles within a radial distance of $4 r_0$. Then
\begin{align*}
    P\{k \geq 2\} &\leq P\{K \geq 2\} \\
    &\xrightarrow[r_0\to0]{} 1 - (1+\rho
    \abs{B(4 r_0)})\exp(-\rho \abs{B(4 r_0)}) \\
    &= 0,
\end{align*}
where $\abs{B(r)}$ denotes the volume of a ball with radius $r$.
It follows that
\begin{align*}
    P\{\mathcal{E} = 0\} &\geq P\{k < 2\} \nonumber \\
                         &= 1-P\{k\geq2\} \nonumber \\
                         &=\xrightarrow[r_0\to0]{} 1.
\end{align*}

Thus the behavior of $k$ by itself implies that the XDF of an ideal gas particle is a single point mass at zero:
\begin{equation}
    p_\mathcal{E}(\varepsilon) =
    \begin{cases}
        1 & \varepsilon = 0 \\
        0 & \varepsilon > 0.
    \end{cases}
\end{equation}
This is consistent with the traditional understanding that the ideal gas is completely devoid of order and structure.

\subsection{Perfect crystal}
\label{subsec:perfect-crystal}

Often held in contrast with the ideal gas is the perfect crystal, in which strong interparticle interactions give rise to a periodic arrangement of equilibrium positions about which particles undergo thermal motion. At zero temperature \footnote{Neglecting zero-point motion}, a perfect crystal with one crystallographically independent site will have an XDF that is degenerate at a certain value $\mathcal{E}_0$. Below we calculate this value for six such crystals, consider the qualitative behavior of the first two XDF moments at positive temperatures, and briefly discuss the treatment of temperature gradients.

\subsubsection{Zero temperature}
\label{subsubsec:zero-temperature}

The dense configuration that results from the stacking of hexagonal layers in an ABC sequence is the face-centered cubic (fcc) crystal structure. Such a structure attains the maximum packing fraction $\varphi$ \cite{Hales2005} and the largest point group $G$ of any periodic arrangement in Euclidean $3$-space \cite{Kaxiras2019}. Its $\mathcal{E}_0$ is calculated at $4.221$ bits.

By instead stacking hexagonal layers in a truncated AB sequence, one obtains the hexagonal close-packed (hcp) crystal structure. Hcp is equal in packing fraction to fcc but has a point group that is only half as large. Repeating the calculation for hcp, it is found that the latter distinction is reflected in a smaller $\mathcal{E}_0$ of $3.835$ bits.

If in addition to truncating the stacking sequence to AB, one uses square rather than hexagonal layers, a body-centered cubic (bcc) crystal structure is obtained. Bcc has the same octahedral point group as fcc but packs particles less efficiently. Evaluating the parameter for bcc, taking $k=14$ as most often done \cite{Frank1958a, Rintoul1996, Ackland2006, Thompson2015, Hwang2019}, we find that a lower packing fraction is likewise reflected in a smaller $\mathcal{E}_0$ as compared to fcc at $4.053$ bits.

Stacking square layers in the trivial AA sequence gives rise to a simple cubic (sc) crystal structure. Sc also possesses the octahedral point symmetry seen in fcc and bcc but has the lowest packing fraction of the three structures. This additional reduction in packing efficiency is observed to produce a still smaller $\mathcal{E}_0$ of $3.185$ bits. 

Table~\ref{tab:crystal-structures} compares the above four crystal structures and two others. Across all pairs of structures it is observed that every increase in $\mathcal{E}_0$ is explained by a corresponding increase in packing fraction or point group order. By contrast the widely used Steinhardt rotational invariant $q_6$ \cite{Steinhardt1983} lacks this property for five of those pairs, namely ($\alpha$-As, sc), (dc, sc), (dc, hcp), (dc, bcc), and (dc, fcc). It is noteworthy that the related invariants $q_4$ and $w_6$ are found upon further examination to be so afflicted in at least as many instances.

\renewcommand{\arraystretch}{1.1}
\begin{table}[b]
    \centering
    \begin{ruledtabular}
        \begin{tabular}{lrrrrrr}
            \noalign{\vskip 2pt}
            Crystal structure
            & \multicolumn{1}{c}{$\varphi$}
            & \multicolumn{1}{c}{$\abs{G}$}
            & \multicolumn{1}{c}{$q_6$}
            & \multicolumn{1}{c}{$k$}
            & \multicolumn{1}{c}{$H(\theta)$}
            & \multicolumn{1}{c}{$\mathcal{E}_0$}
            \\
            \noalign{\vskip 3pt}
            \hline
            \noalign{\vskip 4pt}
            Rhombohedral ($\alpha$-As)   & $0.39$ & $12$ & $0.42$ &  $6$ & $1.371$ & $2.536$ \\
            Diamond cubic (dc)           & $0.34$ & $24$ & $0.62$ &  $3$ & $0.000$ & $2.585$ \\
            Simple cubic (sc)            & $0.52$ & $48$ & $0.35$ &  $6$ & $0.722$ & $3.185$ \\
            Hexagonal close-packed (hcp) & $0.74$ & $24$ & $0.49$ & $12$ & $2.209$ & $3.835$ \\
            Body-centered cubic (bcc)    & $0.68$ & $48$ & $0.50$ & $14$ & $2.455$ & $4.053$ \\
            Face-centered cubic (fcc)    & $0.74$ & $48$ & $0.57$ & $12$ & $1.823$ & $4.221$ \\
            \noalign{\vskip 1pt}
        \end{tabular}
    \end{ruledtabular}
    \caption{The packing fraction $\varphi$, point group order $\abs{G}$, rotational invariant $q_6$, coordination number $k$, bond angle entropy $H(\theta)$, and extracopularity value $\mathcal{E}_0$ for six crystal structures with one crystallographically independent site.}
    \label{tab:crystal-structures}
\end{table}

The results of Sec.~\ref{subsec:quantifying-order} suggest that the local degree of order grows with an increase or control of the bond count $k$ and a control or decrease of the bond angle entropy $H(\theta)$. The interplay between more bonding and less bond angle diversity can be made explicit by rewriting the formula of Lemma~\ref{lem:bayes} so that $k$ and $H(\theta)$ appear in separate factors:
\begin{equation}
    \mathcal{E} = \log_2\bigl(\tbinom{k}{2} \, 2^{-H(\theta)} \bigr).
    \label{eq:rectangular}
\end{equation}
In so doing we recast $\mathcal{E}$ as the logarithm of the area of a rectangle with height $\binom{k}{2}$ and width $2^{-{H(\theta)}}$. Fig.~\ref{fig:rectangular-areas} illustrates this geometric interpretation for five of the crystal structures in Table~\ref{tab:crystal-structures}.

A perfect icosahedral crystal, while incompatible with translational symmetry in Euclidean space, can nevertheless be realized in a space of constant positive curvature \cite{Turci2017, Nelson2002}. The optimality of icosahedral clusters in symmetry, energy, and packing leads us to ask whether a crystal composed of such clusters would have the largest $\mathcal{E}_0$. We capture the implied frontier by setting Eq.~\eqref{eq:rectangular} equal to the value of $\mathcal{E}$ obtained in Example~\ref{example:ico}. This produces the hyperbola seen in Fig.~\ref{fig:rectangular-areas} and given as follows: 
\begin{equation}
    \tbinom{k}{2} \, 2^{-H(\theta)} = 6 \cdot 5^{10/11}, \quad k \geq 2, \quad H(\theta) \geq 0.
\label{eq:hyperbola}
\end{equation}

By Proposition~\ref{prop:icosahedral} the sharp lower bound $E$ of the parameter $\mathcal{E}$ is maximized by the regular icosahedral geometry for spherical bond sets in $D=3$ with $k \leq 12$. And in Table~\ref{tab:geometries-1} we see that $E$ and $\mathcal{E}$ induce identical orderings for the geometries therein considered. It is therefore reasonable to conjecture that $\mathcal{E}$ too is maximized by the icosahedral arrangement under those same conditions.

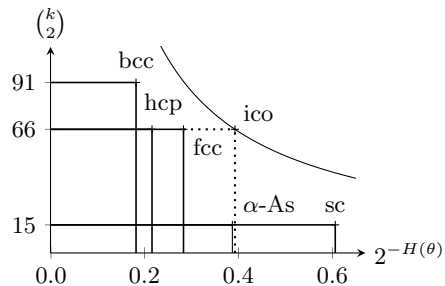
\begin{figure}[t]
    \centering
    \begin{tikzpicture}
        \begin{axis}[
            width=0.667\columnwidth,
            height=0.500\columnwidth,
            xmin=0,
            xmax=0.67,
            ymin=0,
            ymax=110,
            xtick={0.001, 0.2, 0.4, 0.6},
            xticklabels={$0.0$, $0.2$, $0.4$, $0.6$},
            ytick={15, 66, 91},
            xlabel={$2^{-H(\theta)}$},
            ylabel={$\binom{k}{2}$},
            axis lines=middle,
            axis line style={->,>=stealth},
            xlabel style={at={(axis description cs:1,0)}, anchor=west},
            ylabel style={at={(axis description cs:0,1)}, anchor=south},
            domain=0:0.65,
            samples=200
        ]
            \addplot[black] {6*5^(10.0/11.0)/x};
            
            \draw[black, semithick] (axis cs:0,91) -- (axis cs:0.182,91);
            \draw[black, semithick] (axis cs:0.182,0) -- (axis cs:0.182,91);
        
            \draw[black, semithick] (axis cs:0,66) -- (axis cs:0.283,66);
            \draw[black, semithick] (axis cs:0.283,0) -- (axis cs:0.283,66);
        
            \draw[black, semithick] (axis cs:0,66) -- (axis cs:0.216,66);
            \draw[black, semithick] (axis cs:0.216,0) -- (axis cs:0.216,66);
        
            \draw[black, semithick] (axis cs:0,15) -- (axis cs:0.606,15);
            \draw[black, semithick] (axis cs:0.606,0) -- (axis cs:0.606,15);
        
            \draw[black, semithick] (axis cs:0,15) -- (axis cs:0.387,15);
            \draw[black] (axis cs:0.387,0) -- (axis cs:0.387,15);
        
            \draw[black, thick, dotted] (axis cs:0.283,66) -- (axis cs:0.3925,66);
            \draw[black, thick, dotted] (axis cs:0.3925,0) -- (axis cs:0.3925,66);
        
            \addplot[
                only marks,
                mark=+,
                mark size=1.5pt,
                fill=black
            ]
            coordinates {
                (0.182,91)
                (0.283,66)
                (0.216,66)
                (0.606,15)
                (0.387,15)
                (0.3925,66)
            };
        
            \node[anchor=south] at (axis cs:0.182,94) {bcc};
            \node[anchor=west]  at (axis cs:0.285,56.5) {fcc};
            \node[anchor=south] at (axis cs:0.24,69)  {hcp};
            \node[anchor=south west] at (axis cs:0.3925,66) {ico};
            \node[anchor=south] at (axis cs:0.606,17) {sc};
            \node[anchor=south] at (axis cs:0.465,17) {$\alpha$-As};
        \end{axis}
    \end{tikzpicture}
    \caption{The extracopularity parameter, interpreted geometrically as the logarithm of the area of a rectangle. Crystal structures are compared to the icosahedral value and its ``isoextracopularity'' curve. The rectangle for dc is omitted, as it exceeds the depicted range, with its upper right vertex at ${(1,6)}$.}
    \label{fig:rectangular-areas}
\end{figure}

\subsubsection{Positive temperatures}
\label{subsubsec:positive-temperatures}

On account of thermal motion, even a perfect crystal with a single crystallographically independent site will exhibit nonzero variation in $\mathcal{E}$ at positive temperature. Since in any such crystal each atom has a constant $k \geq 2$, the XDF simplifies to
\begin{align}
    p_\mathcal{E}(\varepsilon) &= P\bigl\{ H(\theta) = \log_2\mathopen{}\tbinom{k}{2} - \varepsilon \bigr\} \nonumber.
\end{align}
Using Lemma~\ref{lem:bayes} and Bienaym\'{e}'s identity gives
\begin{align*}
    \operatorname{Var}(\mathcal{E}) &= 
    \operatorname{Var}\bigl(\log_2\mathopen{}\tbinom{k}{2}\bigr) + \operatorname{Var}(H(\theta)) \\
    &\qquad \qquad -2\operatorname{Cov}\bigl(\log_2\mathopen{}\tbinom{k}{2}, H(\theta)\bigr)  \\
    &= \operatorname{Var}(H(\theta)).
\end{align*}
It is now readily seen that for $T>0$, 
\begin{align}
    \operatorname{Var}(\mathcal{E}_T) &= \operatorname{Var}(H(\theta_T)) \nonumber \\
    &> 0 \nonumber \\
    &= \operatorname{Var}(H(\theta_0)) \nonumber \\
    &= \operatorname{Var}(\mathcal{E}_0).
\end{align}

One expects to see the opposite behavior in the mean. It follows from Lemma~\ref{lem:bayes} and the linearity of the expected value that
\begin{equation}
    \langle \mathcal{E} \rangle = \bigl\langle \log_2\mathopen{}\tbinom{k}{2} \bigr\rangle - \langle H(\theta) \rangle.
    \label{eq:expected-extracopularity}
\end{equation}
We expand the ensemble average of the microscopic bond angle entropy as follows:
\begin{equation*}
    \langle H(\theta)\rangle = H(\Theta) - \frac{m-1}{2\binom{k}{2} \ln 2} - \chi(k, m, f_\Theta),
    \label{eq:bond-angle-entropy-expansion}
\end{equation*}
where the first term is the entropy of the discretized ensemble bond angle distribution [Eq.~\eqref{eq:disc-ensemble-bond-angle-distrib}], the second is the first-order bias correction under independent angles \cite{Miller1955}, and the third encompasses higher-order corrections for bias and angle dependence per Eq.~\eqref{eq:joint-ensemble-bond-angle-distrib}.
For a uniform $M$-bin discretization one has
\begin{equation*}
    H(\Theta) = \frac{h(\Theta)}{\ln 2} + \log_2 \frac{180^\circ}{M} + o(1).
\end{equation*}
Call the remainder
\begin{equation*}
    r = H(\Theta) - \biggl(\frac{h(\Theta)}{\ln 2} + \log_2 \frac{180^\circ}{M}\biggr).
\end{equation*}
Inserting the above into Eq.~\eqref{eq:expected-extracopularity} and differentiating the result with respect to temperature yields
\begin{equation}
    \frac{\partial \langle \mathcal{E} \rangle}{\partial T} = -\frac{1}{\ln2}\frac{\partial h(\Theta)}{\partial T} - \frac{\partial r}{\partial T} + \frac{\partial \chi}{\partial T},
\end{equation}
where we have used the fact that $k$, $m$, and $M$ are here constant. A crystalline bond angle distribution consists of finitely many peaks broadening slowly with temperature, so that
\begin{equation*}
    \frac{\partial r}{\partial T} \approx 0 \quad \text{and} \quad \frac{\partial \chi}{\partial T} \approx 0.
\end{equation*}
Applying these and recalling Ineq.~\eqref{ineq:temperature-dependence} yields
\begin{align}
    \frac{\partial \langle \mathcal{E} \rangle}{\partial T} &\approx -\frac{1}{\ln2}\frac{\partial h(\Theta)}{\partial T} \nonumber \\ &<0,
\end{align}
which agrees with our intuition.

\subsubsection{Temperature gradient}
\label{subsubsec:temperature-gradient}

Crystals as well as solids of other kinds are often out of thermal and indeed thermodynamic equilibrium. To facilitate the study of structure in nonequilibrium states we may write the extracopularity parameter as a scalar field $x \mapsto \mathcal{E}(x)$ defined by
\begin{equation*}
\mathcal{E}(x) = \sum_{i=1}^N \mathcal{E}_i \delta(x - x_i).
\end{equation*}

The spatial variation in the local degree of order at a specified length scale $\sigma$ can then be resolved by coarse-graining the microscopic field $\mathcal{E}$ through its convolution with the $D$-dimensional Gaussian kernel
\begin{equation*}
g_\sigma(x) = \frac{1}{(2\pi \sigma^2)^{D/2}} \exp\biggl(-\frac{\abs{x}^2}{2\sigma^2}\biggr).
\end{equation*}
The resulting coarse-grained field is given by
\begin{align*}
\mathcal{E}_\sigma(x) &= (g_\sigma * \mathcal{E})(x) \nonumber \\
&= \int g_\sigma(x - x') \mathcal{E}(x') dx' \nonumber \\
&= \sum_{i=1}^N \mathcal{E}_i g_\sigma(x - x_i).
\end{align*}
An example of this field at the scale of the lattice constant is provided in Fig.~\ref{fig:radial-temp-gradient}.

\begin{figure}[b]
    \centering
    \includegraphics[width=\columnwidth]{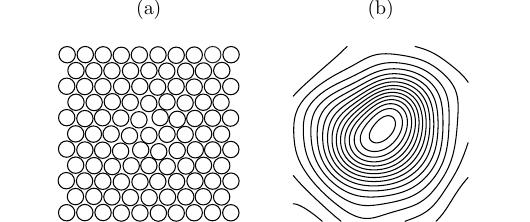}
    \caption{Extracopularity field in two dimensions: (a) a hexagonal crystal with a negative radial temperature gradient, decreasing outward from $T_1 \approx T_m$ at the center to $T_2 \approx 0$ at the boundary of the square region; (b) the coarse-grained extracopularity field $\mathcal{E}_\sigma$ with scale $\sigma$ equal to the lattice constant, decreasing inward from the zero-temperature value of $\mathcal{E}_0 \approx 2.385$ bits at the boundary in increments of $0.05$ bits.}
    \label{fig:radial-temp-gradient}
\end{figure}

\subsection{Simple liquid}
\label{subsec:simple-liquid}
 
The study of the structure of matter in the liquid state is made difficult by the fact such systems admit neither the idealization of independent particle positions nor that of periodic particle arrangement. To address this difficulty, various notions of a ``simple liquid'' have been put forward \cite{Callen1985, Hansen2013, Ingebrigtsen2012}, all centering on a homogeneous condensed fluid with isotropic interactions. In the discussion that follows we consider bulk liquid argon near its triple point, which has long been recognized to exemplify that description.

\subsubsection{Coordination process}
\label{subsubsec:coordination-process}

We begin by conceiving of coordination as a local process in which nearby atoms compete to lie within the vicinity of a reference particle. Regarding this process as sampling without replacement is equivalent to modeling the coordination number $k$ to be hypergeometrically distributed:
\begin{multline}
    \hat{p}_k(\kappa; \Lambda, \lambda, K) = \frac{\binom{\lambda}{\kappa}\binom{\Lambda-\lambda}{K-\kappa}}{{\binom{\Lambda}{K}}}, \\ \max\{0, K+\lambda-\Lambda\} \leq \kappa \leq \min\{K, \lambda\};
    \label{eq:hypergeometric-pmf}
\end{multline}
where the $\Lambda$ (population size), $\lambda$ (number of possible successes), and $K$ (number of draws) are positive integers satisfying $\lambda \leq \Lambda$ and $K \leq \Lambda$ \cite{Greene2017}. Each of these parameters has a clear physical interpretation in the present setting aside from $K$, which we later obtain by using the first moment formula
\begin{align}
    \mu_k &= \mathbb{E}[k] \nonumber \\
    &=\frac{\lambda}{\Lambda} K.
\label{eq:hypergeometric-mean}
\end{align}

The number $\lambda$ of possible successes captures the probabilistic upper bound imposed on the coordination number by packing constraints. For the purposes of determining its value we treat atoms as hard spheres.

Let $r_0$ be half the radial distance to the first RDF maximum (particle radius), $r_1$ the radial distance to the first RDF minimum (interaction radius), and $r_2$ that to its second minimum (outer boundary), each illustrated in Fig.~\ref{fig:hard-sphere-vicinity}(a). Then the volume of the vicinity may be approximated as
\begin{align*}
    V &= \frac{4}{3} \pi (r_1^3 - r_0^3) \nonumber \\
      &\approx \frac{4}{3} \pi [(3r_0)^3 - r_0^3] \nonumber \\
      &= \frac{104}{3} \pi r_0^3.
\end{align*}

\begin{figure}[b]
    \centering
    \includegraphics[width=\columnwidth]{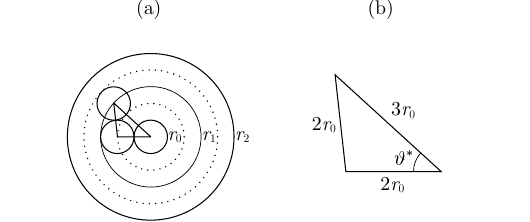}
    \caption{Atomic vicinity in the hard-sphere treatment: (a) the boundaries of key regions; (b) the geometry of the minimum bond angle configuration of an atom and two of its neighbors.}
    \label{fig:hard-sphere-vicinity}
\end{figure}

The maximum packing fraction of an irregular arrangement of identical hard spheres is widely reported to be $\varphi \approx 0.64$  \cite{Bernal1960c,Scott1969,Rintoul1996,Zaccone2022}. The value of $\lambda$ implied by this fraction is determined, up to rounding, by the ratio of the volume of all atoms that lie in the vicinity to that of a single such atom:

\begin{equation}
    \lambda = \biggl\lfloor \frac{ \varphi V}{(4/3)\pi r_0^3} \biggr\rceil = 17,
    \label{eq:num-possible-successes}
\end{equation}
where $\lfloor x\rceil$ denotes the integer nearest to $x$.

The first moment $\mu_k$ corresponds to the average number of atoms within the interaction radius $r_1$. We take the population size $\Lambda$ to be that number within the outer boundary $r_2$. At a temperature of $85$ K and a pressure of $1$ atm \cite{Yarnell1973, Hansen2013, Anikeenko2018} this gives
\begin{equation}
    \begin{gathered}
    \mu_k = 4\pi \rho \int_{r_0}^{r_1} r^2 g(r) dr \approx 12.2, \\
    \Lambda = 4\pi \rho \int_{r_0}^{r_2} r^2 g(r) dr \approx 54.
    \end{gathered}
    \label{eq:rdf-integrals}
\end{equation}
These two values coincide to the nearest integer with the numbers of atoms in the first and first two shells of the Mackay icosahedral cluster, respectively \cite{Mackay1962,Martin1996,Canestrari2025}. 

Now solving Eq.~\eqref{eq:hypergeometric-mean} for $K$ with $\lambda$, $\mu_k$, and $\Lambda$ as given above we obtain
\begin{equation}
    K = \biggl\lfloor \frac{\Lambda}{\lambda} \mu_k \biggr\rceil = 39.
    \label{eq:num-draws}
\end{equation}

Inserting Eqs.~\eqref{eq:num-possible-successes}--\eqref{eq:num-draws} into  Eq.~\eqref{eq:hypergeometric-pmf} produces the following model of the coordination number distribution:
\begin{equation}
    \hat{p}_k(\kappa) = 
    \begin{cases}
        \dfrac {\binom{17}{\kappa}\binom{37}{39-\kappa}}{\binom{54}{39}} & 2 \leq \kappa \leq 17 \\
        0 & \text{otherwise}.
    \end{cases}
    \label{eq:binomial}
\end{equation}

\subsubsection{Bond angle model}
\label{subsubsec:bond-angle-model}

Liquid argon is distinguished from an ideal gas of the same number density primarily by the presence of strong repulsive forces \cite{Weeks1971}. Such forces discourage small interatomic distances, thereby suppressing the occurrence of sharp angles. A basic model of the ensemble bond angle distribution may therefore be derived by truncating the ideal gas solution in Eq.~\eqref{eq:ideal-gas-entropy} as follows:
\begin{equation}
    \hat{p}_\Theta(\vartheta_i) =
    \begin{cases}
        0 & \vartheta_i < \vartheta^* \\
        \sin\vartheta_i / \hat{Z}_\Theta & \vartheta_i \geq \vartheta^*,
    \end{cases}
    \label{eq:ensemble-bond-angle-distribution}
\end{equation}
where $\vartheta^*$ is a cutoff angle, $\hat{Z}_\Theta$ is the normalizer, and the bin-center angles $\vartheta_i$ are given for integers $1 \leq i \leq M$ by
\begin{equation}
    \vartheta_i = 180^\circ \biggl(\frac{2i-1}{2M}\biggr).
    \label{eq:discretization-2}
\end{equation}

The cutoff angle $\vartheta^*$ is determined by appealing to the hard-sphere treatment a second time. As illustrated in Fig.~\ref{fig:hard-sphere-vicinity}(a) and (b), hard spheres produce the smallest attainable bond angle when the central atom combines with two other atoms within a distance of $r_1$ to form an isosceles triangle of side lengths $a=b=2r_0$ and $c=3r_0$. By elementary trigonometry the cosine of its obtuse angle is computed as
\begin{equation*}
    \frac{a^2 + b^2 - c^2}{2ab} = -\frac{1}{8}.
\end{equation*}
Solving for the acute angle yields
\begin{equation}
    \vartheta^* = \frac{180^\circ-\arccos(1/8)}{2} \approx 41.41^\circ.
\end{equation}
Fig.~\ref{fig:truncated-sine-model} depicts the truncation at this angle.

We determine the number $M$ of bins by considering the scale of bond angle fluctuations. Recall that the Lindemann parameter $L$ gives the root-mean-square value of atomic displacements relative to the nearest-neighbor distance at the melting point \cite{Lindemann1910, Lunkenheimer2023}. The corresponding bond angle fluctuations are hence of order
\begin{align*}
    \Delta \theta &= \sqrt{L^2 + L^2} \\
                  &= \sqrt{2}L~\text{rad}.
\end{align*}

It will be seen that $\Delta \theta$ serves as a lower bound on the scale of those same fluctuations in liquid state. Requiring the width of each bin to equal $\Delta \theta$ is therefore tantamount to ensuring that the typical fluctuation of the average bond angle class is registered. This requirement is captured by the function
\begin{align*}
    F : L \mapsto \frac{\pi}{\sqrt{2}L}.
\end{align*}
Reported values of $L$ lie in the range $0.10$--$0.15$ \cite{Saija2006, Oettel2010, Novikov2013}. We assign to $M$ the integer closest to the average of the values of $F$ at the minimum and maximum of this range:
\begin{equation}
    M = \biggl\lfloor\frac{F(0.10) + F(0.15)}{2}\biggr\rceil = 19.
    \label{eq:bin-count}
\end{equation}

Using Eq.~\eqref{eq:bin-count} in Eq.~\eqref{eq:discretization-2} it is found that the inequality $\vartheta_i \geq \vartheta^* \approx 41.41^\circ$ is satisfied for indexes $i \geq 5$. The resulting bond angle model is given in closed form by
\begin{align}
    \hat{p}_\Theta(\vartheta_i) &= \frac{1}{\hat{Z}_\Theta}\sin\vartheta_i, \nonumber \\
    \hat{Z}_\Theta 
    &= \sum_{i=5}^{19} \sin\vartheta_i \approx 10.832; \nonumber \\
    \vartheta_i 
    &= 180^\circ \biggl( \frac{2i-1}{38} \biggr), \nonumber \\[1ex]
    i &= 5, 6, \dots, 19. \label{eq:bond-angle-model}
\end{align}

\begin{figure}[t]
    \centering
    \begin{tikzpicture}
        \begin{axis}[
            width=0.5\columnwidth,
            height=0.4\columnwidth,
            xmin=0,
            xmax=195,
            ymin=0,
            ymax=1.2,
            xtick={0, 41.41, 180},
            xticklabels={0, $41.41^\circ$, $180^\circ$},
            ytick=\empty,
            xlabel={$\vartheta$},
            ylabel={$p_\theta(\vartheta)$},
            samples=720,
            axis lines=middle,
            axis line style={->,>=stealth},
            xlabel style={at={(axis description cs:1,0.00)}, anchor=west},
            ylabel style={at={(axis description cs:0,0.98)}, anchor=south}
        ]
            \addplot[black, very thick, domain=0:41.41] {0};
            \addplot[black, semithick, domain=41.41:180] {sin(x)};
            \draw[black, thick, dotted]
                (axis cs:41.41,0) -- (axis cs:41.41,{sin(41.41)});
            \addplot[only marks, mark=*, mark size=1.6pt, fill=black]
                coordinates {(41.41,{sin(41.41)})};
            \addplot[only marks, mark=o, mark size=1.6pt, line width=0.5pt]
                coordinates {(41.41,0)};
        \end{axis}
    \end{tikzpicture}
    \caption{The left-truncated sine model of the ensemble bond angle distribution for bulk liquid argon.}
    \label{fig:truncated-sine-model}
\end{figure}
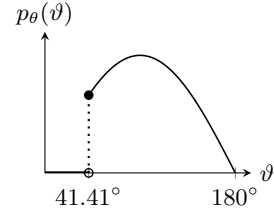

\subsubsection{XDF}
\label{subsubsec:xdf}

Combining the formula obtained in Sec.~\ref{subsec:distribution-function} with the models of the coordination number and bond angle distribution developed in Secs.~\ref{subsubsec:coordination-process} and \ref{subsubsec:bond-angle-model} respectively, we arrive at an expression for the unnormalized XDF of an atom in bulk liquid argon at $85$ K and $1$ atm:
\begin{multline}
    \hat{Z}_\mathcal{E}\hat{p}_\mathcal{E}(\varepsilon) = \boldsymbol{1}\{\varepsilon = 0\}\sum_{\kappa = 0}^2 p_k(\kappa) \\[1ex]
    +\frac{1}{\binom{54}{39}}\sum_{\kappa = 3}^{17} \frac{\tbinom{17}{\kappa}\tbinom{37}{39-\kappa}}{\mathcal{S}(\kappa)} \exp\biggl\{{\biggl(\frac{\mathcal{H}(\varepsilon, \kappa)-\mathcal{M}(\kappa)}{\sqrt{2} \,\mathcal{S}(\kappa)}\!\biggr)^2}\biggr\}, \\[2ex]
    \varepsilon \in \bigl\{ \log_2\mathopen{}\tbinom{a}{2} - b : a \in \mathbb{N}, 2 \leq a \leq 17,  b \in \mathbb{H}_a \bigr\};
    \label{eq:xdf-argon}
\end{multline}
where the set $\mathbb{H}_a$ is as defined in Eq.~\eqref{eq:H-set} and the auxiliary functions $\mathcal{H}$, $\mathcal{M}$, and $\mathcal{S}$ are given respectively by
\begin{align*}
    \mathcal{H}(\varepsilon, \kappa) &= \log_2\mathopen{}\tbinom{\kappa}{2} - \varepsilon, \nonumber \\[2ex]
    \mathcal{M}(\kappa) &= \hat{\mu}_{H(\theta)|k}(\kappa) \\[0.5ex]
    &\approx\max\biggr\{0, \, 3.771 - \frac{14}{\kappa(\kappa-1)\ln2} \biggr\}, \nonumber \\
    \mathcal{S}(\kappa) &= \hat{\sigma}_{H(\theta)|k}(\kappa) \\ 
    &\approx\sqrt{\frac{0.123}{\kappa(\kappa-1)}}. 
\end{align*}
This expression is plotted in Fig.~\ref{fig:liquid-argon-xdf}.

\begin{figure}[b]
    \centering
    \begin{tikzpicture}
        \begin{axis}[
            width=1.0\columnwidth,
            height=0.4\columnwidth,
            xmin=1.3, 
            xmax=3.6,
            ymin=0, 
            ymax=.021,
            xtick={1.0, 1.5, 2.0, 2.5, 3.0, 3.5, 4.0},
            xticklabels={1.0, 1.5, 2.0, 2.5, 3.0, 3.5, 4.0},
            ytick=\empty,
            minor x tick num=4,
            xlabel={$\varepsilon$},
            ylabel={$\hat{Z}_\mathcal{E} \hat{p}_\mathcal{E}(\varepsilon)$},
            axis lines=middle,
            axis line style={->,>=stealth},
            xlabel style={at={(axis description cs:1,0)}, anchor=west},
            ylabel style={at={(axis description cs:0,1)}, anchor=south}
        ]
            \addplot[
                 mark=none,
                 semithick
            ]
            coordinates {
                (0,0)
                (0.01,0)
                (0.02,0)
                (0.03,0)
                (0.04,0)
                (0.05,0)
                (0.06,0)
                (0.07,0)
                (0.08,0)
                (0.09,0)
                (0.1,0)
                (0.11,0)
                (0.12,0)
                (0.13,0)
                (0.14,0)
                (0.15,0)
                (0.16,0)
                (0.17,0)
                (0.18,0)
                (0.19,0)
                (0.2,0)
                (0.21,0)
                (0.22,0)
                (0.23,0)
                (0.24,0)
                (0.25,0)
                (0.26,0)
                (0.27,0)
                (0.28,0)
                (0.29,0)
                (0.3,0)
                (0.31,0)
                (0.32,0)
                (0.33,0)
                (0.34,0)
                (0.35,0)
                (0.36,0)
                (0.37,0)
                (0.38,0)
                (0.39,0)
                (0.4,0)
                (0.41,0)
                (0.42,0)
                (0.43,0)
                (0.44,0)
                (0.45,0)
                (0.46,0)
                (0.47,0)
                (0.48,0)
                (0.49,0)
                (0.5,0)
                (0.51,0)
                (0.52,0)
                (0.53,0)
                (0.54,0)
                (0.55,0)
                (0.56,0)
                (0.57,0)
                (0.58,0)
                (0.59,0)
                (0.6,0)
                (0.61,0)
                (0.62,0)
                (0.63,0)
                (0.64,0)
                (0.65,0)
                (0.66,0)
                (0.67,0)
                (0.68,0)
                (0.69,0)
                (0.7,0)
                (0.71,0)
                (0.72,0)
                (0.73,0)
                (0.74,0)
                (0.75,0)
                (0.76,0)
                (0.77,0)
                (0.78,0)
                (0.79,0)
                (0.8,0)
                (0.81,0)
                (0.82,0)
                (0.83,0)
                (0.84,0)
                (0.85,0)
                (0.86,0)
                (0.87,0)
                (0.88,0)
                (0.89,0)
                (0.9,0)
                (0.91,0)
                (0.92,0)
                (0.93,0)
                (0.94,0)
                (0.95,0)
                (0.96,0)
                (0.97,0)
                (0.98,0)
                (0.99,0)
                (1,0)
                (1.01,0)
                (1.02,0)
                (1.03,0)
                (1.04,0)
                (1.05,0)
                (1.06,0)
                (1.07,0)
                (1.08,0)
                (1.09,0)
                (1.1,0)
                (1.11,0)
                (1.12,0)
                (1.13,0)
                (1.14,0)
                (1.15,0)
                (1.16,0)
                (1.17,0)
                (1.18,0)
                (1.19,0.0001)
                (1.2,0.0001)
                (1.21,0.0001)
                (1.22,0.0001)
                (1.23,0.0001)
                (1.24,0.0001)
                (1.25,0.0001)
                (1.26,0.0001)
                (1.27,0.0001)
                (1.28,0.0001)
                (1.29,0.0002)
                (1.3,0.0002)
                (1.31,0.0002)
                (1.32,0.0002)
                (1.33,0.0002)
                (1.34,0.0002)
                (1.35,0.0002)
                (1.36,0.0003)
                (1.37,0.0003)
                (1.38,0.0003)
                (1.39,0.0003)
                (1.4,0.0003)
                (1.41,0.0003)
                (1.42,0.0003)
                (1.43,0.0003)
                (1.44,0.0003)
                (1.45,0.0003)
                (1.46,0.0003)
                (1.47,0.0003)
                (1.48,0.0003)
                (1.49,0.0003)
                (1.5,0.0003)
                (1.51,0.0003)
                (1.52,0.0003)
                (1.53,0.0004)
                (1.54,0.0004)
                (1.55,0.0005)
                (1.56,0.0006)
                (1.57,0.0006)
                (1.58,0.0007)
                (1.59,0.0008)
                (1.6,0.0009)
                (1.61,0.001)
                (1.62,0.0011)
                (1.63,0.0012)
                (1.64,0.0013)
                (1.65,0.0013)
                (1.66,0.0014)
                (1.67,0.0014)
                (1.68,0.0014)
                (1.69,0.0014)
                (1.7,0.0014)
                (1.71,0.0013)
                (1.72,0.0013)
                (1.73,0.0012)
                (1.74,0.0012)
                (1.75,0.0011)
                (1.76,0.0011)
                (1.77,0.001)
                (1.78,0.0011)
                (1.79,0.0011)
                (1.8,0.0012)
                (1.81,0.0013)
                (1.82,0.0014)
                (1.83,0.0017)
                (1.84,0.0019)
                (1.85,0.0022)
                (1.86,0.0025)
                (1.87,0.0029)
                (1.88,0.0033)
                (1.89,0.0036)
                (1.9,0.0039)
                (1.91,0.0042)
                (1.92,0.0044)
                (1.93,0.0046)
                (1.94,0.0047)
                (1.95,0.0047)
                (1.96,0.0046)
                (1.97,0.0045)
                (1.98,0.0043)
                (1.99,0.004)
                (2,0.0037)
                (2.01,0.0034)
                (2.02,0.0032)
                (2.03,0.0029)
                (2.04,0.0028)
                (2.05,0.0028)
                (2.06,0.0028)
                (2.07,0.003)
                (2.08,0.0034)
                (2.09,0.0039)
                (2.1,0.0045)
                (2.11,0.0052)
                (2.12,0.0061)
                (2.13,0.0069)
                (2.14,0.0078)
                (2.15,0.0087)
                (2.16,0.0094)
                (2.17,0.01)
                (2.18,0.0104)
                (2.19,0.0106)
                (2.2,0.0106)
                (2.21,0.0103)
                (2.22,0.0099)
                (2.23,0.0093)
                (2.24,0.0085)
                (2.25,0.0077)
                (2.26,0.0069)
                (2.27,0.0062)
                (2.28,0.0055)
                (2.29,0.0051)
                (2.3,0.0049)
                (2.31,0.005)
                (2.32,0.0054)
                (2.33,0.006)
                (2.34,0.007)
                (2.35,0.0082)
                (2.36,0.0095)
                (2.37,0.011)
                (2.38,0.0124)
                (2.39,0.0138)
                (2.4,0.0149)
                (2.41,0.0158)
                (2.42,0.0162)
                (2.43,0.0163)
                (2.44,0.0159)
                (2.45,0.0151)
                (2.46,0.014)
                (2.47,0.0127)
                (2.48,0.0113)
                (2.49,0.0098)
                (2.5,0.0085)
                (2.51,0.0073)
                (2.52,0.0065)
                (2.53,0.006)
                (2.54,0.0059)
                (2.55,0.0062)
                (2.56,0.007)
                (2.57,0.0081)
                (2.58,0.0095)
                (2.59,0.011)
                (2.6,0.0126)
                (2.61,0.0141)
                (2.62,0.0154)
                (2.63,0.0163)
                (2.64,0.0167)
                (2.65,0.0166)
                (2.66,0.016)
                (2.67,0.015)
                (2.68,0.0136)
                (2.69,0.012)
                (2.7,0.0103)
                (2.71,0.0086)
                (2.72,0.0072)
                (2.73,0.006)
                (2.74,0.0051)
                (2.75,0.0047)
                (2.76,0.0047)
                (2.77,0.005)
                (2.78,0.0057)
                (2.79,0.0066)
                (2.8,0.0077)
                (2.81,0.0088)
                (2.82,0.0098)
                (2.83,0.0106)
                (2.84,0.0111)
                (2.85,0.0112)
                (2.86,0.0109)
                (2.87,0.0102)
                (2.88,0.0092)
                (2.89,0.008)
                (2.9,0.0068)
                (2.91,0.0056)
                (2.92,0.0045)
                (2.93,0.0035)
                (2.94,0.0029)
                (2.95,0.0024)
                (2.96,0.0023)
                (2.97,0.0024)
                (2.98,0.0026)
                (2.99,0.003)
                (3,0.0035)
                (3.01,0.0039)
                (3.02,0.0043)
                (3.03,0.0046)
                (3.04,0.0046)
                (3.05,0.0045)
                (3.06,0.0042)
                (3.07,0.0038)
                (3.08,0.0032)
                (3.09,0.0027)
                (3.1,0.0021)
                (3.11,0.0016)
                (3.12,0.0012)
                (3.13,0.0009)
                (3.14,0.0008)
                (3.15,0.0007)
                (3.16,0.0006)
                (3.17,0.0007)
                (3.18,0.0008)
                (3.19,0.0009)
                (3.2,0.001)
                (3.21,0.001)
                (3.22,0.0011)
                (3.23,0.001)
                (3.24,0.001)
                (3.25,0.0008)
                (3.26,0.0007)
                (3.27,0.0006)
                (3.28,0.0004)
                (3.29,0.0003)
                (3.3,0.0002)
                (3.31,0.0002)
                (3.32,0.0001)
                (3.33,0.0001)
                (3.34,0.0001)
                (3.35,0.0001)
                (3.36,0.0001)
                (3.37,0.0001)
                (3.38,0.0001)
                (3.39,0.0001)
                (3.4,0.0001)
                (3.41,0.0001)
                (3.42,0.0001)
                (3.43,0.0001)
                (3.44,0.0001)
                (3.45,0)
                (3.46,0)
                (3.47,0)
                (3.48,0)
                (3.49,0)
                (3.5,0)
                (3.51,0)
                (3.52,0)
                (3.53,0)
                (3.54,0)
                (3.55,0)
                (3.56,0)
                (3.57,0)
                (3.58,0)
                (3.59,0)
                (3.6,0)
                (3.61,0)
                (3.62,0)
                (3.63,0)
                (3.64,0)
                (3.65,0)
                (3.66,0)
                (3.67,0)
                (3.68,0)
                (3.69,0)
                (3.7,0)
                (3.71,0)
                (3.72,0)
                (3.73,0)
                (3.74,0)
                (3.75,0)
                (3.76,0)
                (3.77,0)
                (3.78,0)
                (3.79,0)
                (3.8,0)
                (3.81,0)
                (3.82,0)
                (3.83,0)
                (3.84,0)
                (3.85,0)
                (3.86,0)
                (3.87,0)
                (3.88,0)
                (3.89,0)
                (3.9,0)
                (3.91,0)
                (3.92,0)
                (3.93,0)
                (3.94,0)
                (3.95,0)
                (3.96,0)
                (3.97,0)
                (3.98,0)
                (3.99,0)
                (4,0)
                (4.01,0)
                (4.02,0)
                (4.03,0)
                (4.04,0)
                (4.05,0)
                (4.06,0)
                (4.07,0)
                (4.08,0)
                (4.09,0)
                (4.1,0)
                (4.11,0)
                (4.12,0)
                (4.13,0)
                (4.14,0)
                (4.15,0)
                (4.16,0)
                (4.17,0)
                (4.18,0)
                (4.19,0)
                (4.2,0)
                (4.21,0)
                (4.22,0)
                (4.23,0)
                (4.24,0)
                (4.25,0)
                (4.26,0)
                (4.27,0)
                (4.28,0)
                (4.29,0)
                (4.3,0)
                (4.31,0)
                (4.32,0)
                (4.33,0)
                (4.34,0)
                (4.35,0)
                (4.36,0)
                (4.37,0)
                (4.38,0)
                (4.39,0)
                (4.4,0)
                (4.41,0)
                (4.42,0)
                (4.43,0)
                (4.44,0)
                (4.45,0)
                (4.46,0)
                (4.47,0)
                (4.48,0)
                (4.49,0)
                (4.5,0)
                (4.51,0)
                (4.52,0)
                (4.53,0)
                (4.54,0)
                (4.55,0)
                (4.56,0)
                (4.57,0)
                (4.58,0)
                (4.59,0)
                (4.6,0)
                (4.61,0)
                (4.62,0)
                (4.63,0)
                (4.64,0)
                (4.65,0)
                (4.66,0)
                (4.67,0)
                (4.68,0)
                (4.69,0)
                (4.7,0)
                (4.71,0)
                (4.72,0)
                (4.73,0)
                (4.74,0)
                (4.75,0)
                (4.76,0)
                (4.77,0)
                (4.78,0)
                (4.79,0)
                (4.8,0)
                (4.81,0)
                (4.82,0)
                (4.83,0)
                (4.84,0)
                (4.85,0)
                (4.86,0)
                (4.87,0)
                (4.88,0)
                (4.89,0)
                (4.9,0)
                (4.91,0)
                (4.92,0)
                (4.93,0)
                (4.94,0)
                (4.95,0)
                (4.96,0)
                (4.97,0)
                (4.98,0)
                (4.99,0)
                (5,0)
            };
            \node[anchor=south] at (axis cs:1.68, 0.0014) {$9$};
            \node[anchor=south] at (axis cs:1.95, 0.0047) {$10$};
            \node[anchor=south] at (axis cs:2.19, 0.0106) {$11$};
            \node[anchor=south] at (axis cs:2.43, 0.0163) {$12$};
            \node[anchor=south] at (axis cs:2.64, 0.0167) {$13$}; 
            \node[anchor=south] at (axis cs:2.85, 0.0112) {$14$};
            \node[anchor=south] at (axis cs:3.04, 0.0046) {$15$};
            \node[anchor=south] at (axis cs:3.22, 0.0011) {$16$};
            \addplot[
                ycomb,
                mark=none,
                thick, 
                dotted
            ] 
            coordinates {
                (1.68, 0.0014)
                (1.95, 0.0047)
                (2.19, 0.0106)
                (2.43, 0.0163)
                (2.64, 0.0167)
                (2.85, 0.0112)
                (3.04, 0.0046)
                (3.22, 0.0011)
            };
        \end{axis}
    \end{tikzpicture}
    \caption{Unnormalized XDF for liquid argon at $85$ K and $1$ atm, as predicted by $\hat{Z}_\mathcal{E}\hat{p}_\mathcal{E}$. Modal coordination numbers for the eight prominent peaks are indicated above the curve.}
    \label{fig:liquid-argon-xdf}
\end{figure}
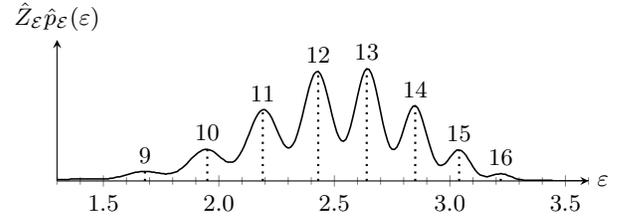

\bibliography{main.bib}

\end{document}